\documentclass[abstracton]{scrartcl}
\pdfoutput=1
\RequirePackage[utf8]{inputenc}
\RequirePackage[T1]{fontenc}

\RequirePackage{enumerate}
\RequirePackage{amsmath}
\RequirePackage{amssymb}

\RequirePackage[hyperindex,breaklinks,pdfborder={0 0 0}]{hyperref}
\usepackage[ style=ieee-alphabetic, firstinits=false ]{biblatex}
\bibliography{references}
\newcommand{\citeN}[2][]{\textcite[#1]{#2}}

\RequirePackage{mathrsfs}
\RequirePackage[amsmath,hyperref,thmmarks]{ntheorem}

\theoremseparator{.}
\newtheorem{theorem}{Theorem}[section]
\newtheorem{definition}[theorem]{Definition}
\newtheorem{lemma}[theorem]{Lemma}
\newtheorem{proposition}[theorem]{Proposition}
\newtheorem{corollary}[theorem]{Corollary}

\theoremstyle{empty}
\renewtheorem{theorem*}{Theorem}
\renewtheorem{lemma*}{Lemma}
\theoremstyle{nonumberplain}
\theoremheaderfont{\normalfont\itshape}
\theorembodyfont{\normalfont}

\theoremsymbol{\ensuremath{_\blacksquare}}
\newtheorem{proof}{Proof}
\RequirePackage{tikz}
\usetikzlibrary{arrows,shapes,backgrounds,calc,decorations.pathreplacing,patterns}
\tikzstyle{bullet}=[draw,semithick,circle,inner sep=0pt,minimum size=5pt]

\tikzstyle{square}=[draw,semithick,rectangle,inner sep=0pt,minimum size=5pt]
\tikzstyle{vertex}=[circle, draw, inner sep=0pt, minimum width=2pt]
\tikzstyle{edge}=[semithick,cap=round]
\tikzstyle{dedge}=[edge,>=latex,->,join=round]
\newcommand{\cc}[1]{\ensuremath{\mathrm{#1}}}
\newcommand{\pp}[1]{\textsc{#1}}
\newcommand{\op}[1]{\ensuremath{\operatorname{#1}}}

\newcommand{\poly}{\op{poly}}

\newcommand{\bbbn}{\mathbb{N}}

\newcommand{\bbbq}{\mathbb{Q}}
\newcommand{\bbbr}{\mathbb{R}}

\newcommand{\N}{\bbbn}

\newcommand{\Q}{\bbbq}
\newcommand{\R}{\bbbr}

\newcommand{\dotcup}{\mathbin{\dot\cup}}

\newcommand{\ETH}{{\upshape ETH}}    % exponential time hypothesis
\newcommand{\rETH}{{\upshape rETH}}  % randomized exponential time hypothesis
\newcommand{\cETH}{{\upshape \#ETH}} % counting exponential time hypothesis

\newcommand{\red}{\preccurlyeq}% reductions

\newcommand{\scr}[1]{\mathscr{#1}}

\newcommand{\perm}{\operatorname{per}}

\newcommand{\comp}{k}

\newcommand{\ZNUL}[3] {\mbox{$Z_0({#1};{#2},{#3})$}} 	% Z_0(G;q,w)
\newcommand{\ZREL}[2]{\mbox{\ZNUL{#1}{0}{#2}}}			% Z_0(G;0,w)
\newcommand{\inflate}[2] {\mbox{${#1}\otimes{#2}$}}		% G x H

\newcommand{\wump}[1] {\mbox{$W_{#1}$}}				% W_S (Wump graph)

\newcommand{\Oh}{O}

\newcommand{\iref}[1]{(\ref{#1})}
\title{%
Exponential Time Complexity\\
of the Permanent and the Tutte Polynomial%
\thanks{%
  The journal version of this paper appears in the ACM Transactions on
  Algorithms~\cite{TALGversion}.
  Preliminary versions appeared in the
  proceedings of ICALP~2010~\cite{ICALPversion}
  and IPEC~2010~\cite{HusfeldtTaslaman}.
}
}
\author{%
  \begin{minipage}{.45\linewidth}
    \begin{center}
      Holger Dell%
\thanks{%
  Research partially supported by the Alexander von Humboldt
  Foundation and NSF grant 1017597.%
}%
      \\\small
      University of Wisconsin--Madison, USA%
      \\\small\tt
      holger@cs.wisc.edu
    \end{center}
  \end{minipage}
\and
  \begin{minipage}{.45\linewidth}
    \begin{center}
      Thore Husfeldt
      \\\small
      IT University of Copenhagen, Denmark
      \\\small
      Lund University, Sweden%
      \\\small\tt
      thore@itu.dk
    \end{center}
  \end{minipage}
\and
  \begin{minipage}{.45\linewidth}
    \begin{center}
      D\'aniel Marx%
\thanks{%
  Research supported by ERC Starting Grant PARAMTIGHT (280152).%
}%
      \\\small
      Computer and Automation Research Institute,
      Hungarian Academy of Sciences (MTA SZTAKI), Budapest, Hungary
      \\\small\tt
      dmarx@cs.bme.hu
    \end{center}
  \end{minipage}
\and
  \begin{minipage}{.45\linewidth}
    \begin{center}
      Nina Taslaman
      \\\small
      IT University of Copenhagen, Denmark%
      \\\small\tt
      nsta@itu.dk
    \end{center}
  \end{minipage}
\and
  \begin{minipage}{.45\linewidth}
    \begin{center}
      Martin Wahl\'en
      \\\small
      Lund University, Sweden
      \\\small
      Uppsala University, Sweden
      \\\small\tt
      mva@df.lth.se
    \end{center}
  \end{minipage}
}

\date{\mbox{}\\[.5\baselineskip]June, 2012\vspace{-1.5\baselineskip}}

\begin{document}

\maketitle

\begin{abstract}
We  show conditional lower bounds for well-studied \cc{\#P}-hard
problems:
\begin{itemize}
  \item
    The number of satisfying assignments of a $2$-CNF formula with $n$
    variables cannot be computed in time $\exp(o(n))$, and the same is
    true for computing the number of all independent sets in an
    $n$-vertex graph.
  \item
    The permanent of an $n\times n$ matrix with entries~$0$ and~$1$
    cannot be computed in time $\exp(o(n))$. %
  \item
    The Tutte polynomial of an $n$-vertex multigraph
    cannot be computed in time $\exp(o(n))$ at most evaluation
    points~$(x,y)$ in the case of multigraphs, and it cannot be
    computed in time $\exp(o(n/\poly\log n))$ in the case of simple
    graphs. %
\end{itemize}
Our lower bounds are relative to (variants of) the Exponential Time
Hypothesis (\ETH{}), which says that the satisfiability of
$n$-variable $3$-CNF formulas cannot be decided in time $\exp(o(n))$. %
We relax this hypothesis by introducing its counting version \cETH{},
namely that the satisfying assignments cannot be counted in time
$\exp(o(n))$. %
In order to use \cETH{} for our lower bounds, we transfer the
sparsification lemma for $d$-CNF formulas to the counting setting. %
\end{abstract}
\section{Introduction}

The permanent of a matrix and the Tutte polynomial of a graph are
central topics in the study of counting algorithms. %
Originally defined in the combinatorics literature, they unify and
abstract many enumeration problems, including immediate questions
about graphs such as computing the number of perfect matchings,
spanning trees, forests, colourings, certain flows and orientations,
but also less obvious connections to other fields, such as link
polynomials from knot theory, reliability polynomials from network
theory, and (maybe most importantly) the Ising and Potts models from
statistical physics.

From its definition (repeated in \eqref{eq: per def} below), the
permanent of an $n\times n$-matrix can be computed in $O(n!n)$ time,
and the Tutte polynomial \eqref{eq: Tutte def} can be evaluated in
time exponential in the number of edges. %
Both problems are famously \cc{\#P}-hard, which rules out the
existence of polynomial-time algorithms under standard
complexity-theoretic assumptions, but that does not mean that we have
to resign ourselves to brute-force evaluation of the definition. %
In fact, Ryser's famous formula~\cite{R63} computes the permanent with
only $\exp(O(n))$ arithmetic operations, and more recently, an
algorithm with running time $\exp(O(n))$ for $n$-vertex graphs has
also been found~\cite{BHKK08} for the Tutte polynomial. %
Curiously, both of these algorithms are based on the
inclusion--exclusion principle. %
In the present paper, we show that these algorithms are not likely to be
significantly improved, by providing conditional lower bounds of
$\exp(\Omega(n))$ for both problems.

It is clear that \#P-hardness is not the right conceptual framework
for such claims, as it is unable to distinguish between different
types of super-polynomial time complexities. %
For example, the Tutte polynomial for planar graphs remains \#P-hard,
but can be computed in time $\exp(O(\sqrt{n}))$ \cite{SIT95}. %
Therefore, we work under Impagliazzo and Paturi's
\emph{Exponential Time Hypothesis}
(\ETH{}), viz. the complexity theoretic assumption that \emph{some}
hard problem, namely the satisfiability of $3$-CNF formulas in $n$
variables, cannot be solved in time $\exp(o(n))$~\cite{IP01}. %
More specifically, we introduce \cETH{}, a counting analogue of \ETH{}
which models the hypothesis that \emph{counting} the satisfying
assignments cannot be done in time $\exp(o(n))$.

\subsection*{Computing the permanent}
The permanent of an $n\times n$ matrix $A$ is defined as
\begin{equation} \label{eq: per def}%
  \perm A = \sum_{\pi \in S_n} \prod_{1\leq i \leq n} A_{i \pi(i)}\,,
\end{equation}
where $S_n$ is the set of permutations of $\{1,\ldots,n\}$. %
This is redolent of the determinant from linear algebra, $\det A=
\sum_\pi \op{sign}(\pi) \prod_i A_{i\pi(i)}$, the only difference is
an easily computable sign for every summand. %
However small this difference in the definition may seem, the
determinant and the permanent are believed to be of a vastly different
computational calibre.
Both definitions involve a summation with $n!$ terms and both problems
have much faster algorithms that are textbook material: The
determinant can be computed in polynomial time using Gaussian
elimination and the permanent can be computed in $O(2^nn)$ operations
using Ryser's formula.
Yet, the determinant seems to be exponentially easier to compute
than the permanent.

Valiant's celebrated \cc{\#P}-hardness result for the
permanent~\cite{V79} shows that no polynomial-time algorithm \`a la
``Gaussian elimination for the permanent'' can exist unless
$\cc{P}=\cc{NP}$, and indeed unless $\cc{P}=\cc{P}^\cc{\#P}$.
Several unconditional lower bounds for the permanent in restricted
models of computation are also known. %
\citeN{JS82} have shown that monotone arithmetic
circuits need $n(2^{n-1}-1)$ multiplications to compute the permanent,
a bound they can match with a variant of Laplace's determinant
expansion. %
\citeN{Raz} has shown that multi-linear arithmetic formulas for the
permanent require size~$\exp(\Omega(\log^2 n))$.  Ryser's formula
belongs to this class of formulas, but is much larger than the lower
bound; no smaller construction is known. %
Intriguingly, the same lower bound holds for the determinant, where it
is matched by a formula of size $\exp(O(\log^2 n))$ due to
\citeN{Ber84}. %
One of the consequences of the present paper is that Ryser's formula
is in some sense optimal under \cETH{}. %
In particular, no uniformly constructible, subexponential size formula
such as Berkowitz's can exist for the permanent unless~\cETH{} fails.

A related topic is the expression of $\perm A$ in terms of $\det
f(A)$, where $f(A)$ is a matrix of constants and entries from $A$ and
is typically much larger than $A$. %
This question has fascinated many mathematicians for a long time, see
Agrawal's survey~\cite{Ag}; the best known bound on the dimension of
$f(A)$ is $\exp(O (n))$ and it is conjectured that all such
constructions require exponential size. %
In particular, it is an important open problem if a permanent of size
$n$ can be expressed as a determinant of size $\exp(O(\log^2 n))$. %
We show that under \cETH{}, if such a matrix $f(A)$ exists, computing
$f$ must take time $\exp(\Omega(n))$.

\subsection*{Computing the Tutte polynomial}

The Tutte polynomial, a bivariate polynomial associated with a given
graph $G=(V,E)$ with $n$ vertices and $m$ edges, is defined as
\begin{equation} \label{eq: Tutte def}
  T(G;x,y) =
  \sum_{A\subseteq E} (x-1)^{\comp (A)-\comp(E)}
  (y-1)^{\comp(A)+|A|-|V|}\,,
\end{equation}
where $\comp(A)$ denotes the number of connected components of the
subgraph $(V,A)$.

% The Tutte polynomial is among the most well-studied objects in
% computational complexity and combinatorics, and we will not attempt
% to survey the literature here (instead,
% cf.~\cite{bk:DongKohTeo2005,mako06zoo}, for instance).  We will not
% here repeat the definitions of these specialisations and their
% connections with \eqref{eq: Tutte def}; this can be found in,
% e.g.,~\cite{JVW90}.

Despite their unified definition \eqref{eq: Tutte def}, the various
computational problems given by $T(G;x,y)$ for different points
$(x,y)$ differ widely in computational complexity, as well as in the
methods used to find algorithms and lower bounds. %
For example, $T(G;1,1)$ equals the number of spanning trees in $G$,
which happens to admit a polynomial-time algorithm, curiously again
based on Gaussian elimination. %
On the other hand, the best known algorithm for computing $T(G;2,1)$,
the number of forests, runs in $\exp(O(n))$ time.

Computation of the Tutte polynomial has fascinated researchers in
computer science and other fields for many decades. %
For example, the algorithms of Onsager and Fischer from the 1940s and
1960s for computing the so-called partition function for the planar
Ising model are viewed as major successes of statistical physics and
theoretical chemistry; this corresponds to computing $T(G;x,y)$ along
the hyperbola $(x-1)(y-1)=2$ for planar~$G$. %
Many serious attempts were made to extend these results to other
hyperbolas or graph classes, but ``after a quarter of a century and
absolutely no progress,'' Feynman in 1972 observed that ``the exact
solution for three dimensions has not yet been found.''%
\footnote{The Feynman quote and many other quotes describing the
  frustration and puzzlement of physicists around that time can be
  found in the copious footnotes of \cite{I00}.}

The failure of theoretical physics to ``solve
the Potts model'' and sundry other questions implicit in the
computational complexity of the Tutte polynomial were explained only
with Valiant's \#P-hardness programme. %
After a number of papers, culminating in the work of \citeN{JVW90}, the
polynomial-time complexity of exactly computing the Tutte polynomial
at points $(x,y)$ is now completely understood: it is \cc{\#P}-hard
everywhere except at those points $(x,y)$ where a polynomial-time
algorithm is known; these points consist of the hyperbola
$(x-1)(y-1)=1$ as well as the four points
$(1,1),(-1,-1),(0,-1),(-1,0)$.

In the present paper, we show an $\exp(\Omega(n))$ lower bound to match
the $\exp(O(n))$ algorithm from~\cite{BHKK08}, which holds under
\cETH{} everywhere except for $|y|=1$. %
In particular, this establishes a gap to the planar case, which admits
an $\exp(O(\sqrt{n}))$ algorithm~\cite{SIT95}. %
Our hardness results apply (though not everywhere, and sometimes with
a weaker bound) even if the graphs are sparse and simple. %
These classes are of particular interest because most of the graphs
arising from applications in statistical mechanics arise from bond
structures, which are sparse and simple.

It has been known since the 1970s \cite{L76} that graph $3$-colouring
can be solved in time $\exp(O(n))$, and this is matched by an
$\exp(\Omega(n))$ lower bound under \ETH{}~\cite{IPZ01}. %
Since graph $3$-colouring corresponds to evaluating~$T$ at $(-2,0)$,
the exponential time complexity for $T(G;-2,0)$ was thereby already
understood. %
In particular, computing $T(G;x,y)$ for input $G$ \emph{and} $(x,y)$
requires vertex-exponential time, an observation that is already made
in~\cite{GHN06} without explicit reference to \ETH{}.

The literature for computing the Tutte polynomial is very rich, and we
make no attempt to survey it here. %
A recent paper of \citeN{GJ08Tutte}, which shows that the Tutte
polynomial is hard to even approximate for large parts of the Tutte
plane, contains an overview. %
A list of graph classes for which subexponential time algorithms are
known can be found in~\cite{BHKK08}.

\subsection*{Complexity assumptions}
\label{sec:compass}
The standard complexity assumption $\cc{P}\neq\cc{NP}$ is not
sufficient for our purposes: it is consistent with current knowledge
that $\cc{P}\neq\cc{NP}$ holds and yet NP-hard problems such as
\pp{$3$-Sat} have subexponential time algorithms.
What we need is a complexity assumption stating that certain problems
can be solved only in exponential time.

The exponential time hypothesis (\ETH{}) by~\citeN{IP01} is
that satisfiability of $3$-CNF formulas cannot be computed
substantially faster than by trying all possible assignments.
Formally, this reads as follows:

\vspace{5pt}
\begin{minipage}{.95\linewidth}
  \hspace{\parindent}
  (\ETH{})
  \quad
  \begin{minipage}{.85\linewidth}
    There is a constant $c>0$ such that no deterministic algorithm
    can decide \pp{$3$-Sat} in time~$\exp(c\cdot n)$.
  \end{minipage}
\end{minipage}
\vspace{5pt}

A different way of formulating \ETH{} is to say that there is no
algorithm deciding \pp{$3$-Sat} in time $\exp(o(n))$.
The latter statement is clearly implied by the above statement, and it
will be more convenient for discussion to use this form and state
results this way.

In two of our lower bounds, Theorem~\ref{thm: 2sat} and
Theorem~\ref{thm: perm}\iref{thmi: perm RETH}, we need a slightly
stronger assumption that rules out the possibility of randomized
algorithms as well:

\vspace{5pt}
\begin{minipage}{.95\linewidth}
  \hspace{\parindent}
  (\rETH{})
  \quad
  \begin{minipage}{.85\linewidth}
    There is a constant $c>0$ such that no {\em randomized} algorithm
    can decide \pp{$3$-Sat} in time~$\exp(c\cdot n)$ with error
    probability at most~$1/3$.
  \end{minipage}
\end{minipage}
\vspace{5pt}

The reason why we need \rETH{} in these two proofs is that we are
reducing from the promise problem \pp{Unique $3$-Sat}, which is
\pp{$3$-Sat} with the promise that the given $3$-CNF formula has at
most one satisfying assignment.
\citeN{Calabro_isolation} established a lower bound on \pp{Unique
$3$-Sat} assuming \rETH{}, thus our results are also relative to this
complexity assumption.
By reducing from \pp{Unique $3$-Sat}, we avoid the use of
interpolation, which typically weakens the lower bound by
polylogarithmic factors in the exponent.

Intuitively, counting the number of solutions is much harder than
deciding the existence of a solution: in the latter case, we only need
to find a single solution, while in the former case we have to somehow
reason about the set of all possible solutions.
A formal evidence is that many natural counting problems are
\cc{\#P}-hard and therefore not only as hard as all problems in
\cc{NP} but as hard as all the problems in the polynomial-time
hierarchy~\cite{toda89}.
If counting problems seem to be so much harder, then it is natural to
ask if their hardness can be demonstrated by a weaker complexity
assumption than what is needed for the decision problems.
We show that our lower bounds, with the exception of Theorem~\ref{thm:
2sat} and Theorem~\ref{thm: perm}\iref{thmi: perm RETH}, can be
obtained using the weaker complexity assumption stating that counting
the number of solutions to a 3-CNF formula requires exponential time
(i.e., a counting variant of \ETH{}).
\begin{quote}\small
  \begin{description}
    \item[Name] \pp{\#$3$-Sat}
    \item[Input] $3$-CNF formula $\varphi$ with $n$ variables and $m$
      clauses.
    \item[Output] The number of satisfying assignments to $\varphi$.
  \end{description}
\end{quote}
The best known algorithm for this problem runs in
time~$O(1.6423^n)$~\cite{Kutzkov20071counting3sat}. %

\vspace{5pt}
\begin{minipage}{.95\linewidth}
  \hspace{\parindent}
  (\cETH{})
  \hfill
  \begin{minipage}{.85\linewidth}
    There is a constant $c>0$ such that no deterministic
    algorithm can compute \pp{\#$3$-Sat} in time~$\exp(c\cdot n)$.
  \end{minipage}
\end{minipage}
\vspace{5pt}

\ETH{} trivially implies \cETH{} whereas the other direction is not
known.

By introducing the sparsification lemma, \citeN{IPZ01} show that \ETH{}
is a robust notion in the sense that the clause width~$3$ and the
parameter~$n$ (number of variables) in its definition can be replaced
by $d\geq 3$ and~$m$ (number of clauses), respectively, to get an
equivalent hypothesis, albeit the constant~$c$ may change in doing
so. As most of the reductions are sensitive to the number of clauses,
this stronger form of \ETH{} is essential for proving tight lower bounds
for concrete problems.  In order to be able to use \cETH{} in such
reductions, we transfer the sparsification lemma to \pp{\#$d$-Sat} and
get a similar kind of robustness for \cETH{}.
\begin{theorem}\label{thm: counting_sparsification_essence}
  Let $d \geq 3$ be an integer. Then
  \cETH{} holds
  if and only if there is a constant $c>0$ such that no deterministic
  algorithm can solve
  \pp{\#$d$-Sat} in time $\exp(c\cdot m)$.
\end{theorem}
The proof of this theorem is spelled out in Appendix~\ref{sec:
  counting sparsification}.
The relationship between \cETH{} and the parameterized complexity of
counting problems is explained in Appendix~\ref{sec: parameterized}.

\subsection*{Results: Counting Independent Sets}

In light of Theorem~\ref{thm: counting_sparsification_essence}, it is
natural to consider the exponential time complexity of \pp{\#$2$-Sat}.
Restricted to antimonotone $2$-CNF formulas, this corresponds to
counting \emph{all} independent sets in a given graph, which
cannot be done in time $\exp(o(n/\log^3 n))$  under \cETH~\cite{Hoffmann2010}.
The loss of the $\poly\log$-factor in the exponent is due to the
interpolation inherent in the hardness reduction.
We avoid interpolation using the isolation lemma for $d$-CNF formulas
by~\citeN{Calabro_isolation}, and we get an asymptotically tight lower
bound.
The drawback is that our lower bound only holds under the randomized
version of \ETH{} instead of \cETH{}.
\begin{theorem}\label{thm: 2sat}%
  Under \rETH{}, there is no randomized algorithm that computes the
  number of all independent sets in time
  $\exp(o(m))$, where $m$ is the number of edges.
  Under the same assumption, there is no randomized algorithm for
  \pp{\#$2$-Sat} that runs in time $\exp(o(m))$, where $m$ is the
  number of clauses.
\end{theorem}
We discuss the isolation technique and prove this theorem in
\S\ref{sec: indsets}.
\subsection*{Results: The Permanent}

For a set $S$ of rationals we define the following problems:
\begin{quote}\small
  \begin{description}
    \item[Name]  $\pp{Perm}^S$
    \item[Input] Square matrix $A$ with entries from $S$.
    \item[Output] The value of $\perm A$.
  \end{description}
\end{quote}
We write $\pp{Perm}$ for $\pp{Perm}^\N$. %
If $B$ is a bipartite graph with $A_{ij}$ edges from the $i$th vertex
in the left half to the $j$th vertex in the right half $(1\leq i,j\leq
n)$, then $\perm (A)$ equals the number of perfect matchings
of~$B$. %
Thus $\pp{Perm}$ and $\pp{Perm}^{0,1}$ can be viewed as counting the
perfect matchings in bipartite multigraphs and bipartite simple
graphs, respectively.
We express our lower bounds in terms of $m$, the number of non-zero
entries of~$A$. %
Without loss of generality, $n \leq m$, so the same bounds hold for
the parameter~$n$ as well. %
\begin{theorem}\mbox{ }
\label{thm: perm}
  \begin{enumerate}[(i)]
    \item\label{thmi: perm standard}%
      $\pp{Perm}^{-1,0,1}$ and $\pp{Perm}$
      cannot be computed in time $\exp(o(m))$ under \cETH.
    \item\label{thmi: perm zeroone}%
      $\pp{Perm}^{0,1}$ cannot be computed in time $\exp(o(m/\log n))$ under \cETH.
    \item\label{thmi: perm RETH}%
      $\pp{Perm}^{0,1}$ cannot be computed in time $\exp(o(m))$ under \rETH.
\end{enumerate}
\end{theorem}
The proof of this theorem is in \S\ref{sec: permanent}. %
For~\iref{thmi: perm standard}, we follow a standard reduction by
Valiant \cite{V79,Papa} but use a simple equality gadget derived
from~\cite{BD07} instead of Valiant's XOR-gadget, and we use
interpolation to get rid of the negative weights. %
To establish~\iref{thmi: perm zeroone} we simulate edge weights $w>1$
by gadgets of size logarithmic in $w$, which increases the number of
vertices and edges by a logarithmic factor. %
For~\iref{thmi: perm RETH} we use the isolation lemma and the reduction
from part \iref{thmi: perm standard}, and we simulate the edge
weights~$-1$ without interpolation by replacing them with~$2$ and
doing computation modulo~$3$.  Observe that~\iref{thmi: perm RETH}
is an asymptotically tight lower bound while~\iref{thmi: perm zeroone}
is not, but it also uses the stronger complexity assumption \rETH{}
instead of \cETH.

\subsection*{Results: The Tutte Polynomial}

The computational problem $\pp{Tutte}(x,y)$ is defined for each pair
$(x,y)$ of rationals.
\begin{quote}\small
\begin{description}
  \item[Name] $\pp{Tutte}(x,y)$.
  \item[Input] Undirected multigraph $G$ with $n$ vertices.
  \item[Output] The value of $T(G;x,y)$.
\end{description}\end{quote}
In general, parallel edges and loops are allowed; we write
$\pp{Tutte}^{0,1}(x,y)$ for the special case where the input graph is
simple.

Our main result is that, under \cETH{}, $\pp{Tutte}(x,y)$ cannot be
computed in time $\exp(o(n))$ for specific points $(x,y)$.
However, the size of the bound, and the graph classes for which it
holds, varies. %
We summarise our results in the theorem below, see also
Figure~\ref{fig: Tutte}. %
\newcommand\legend[2]{%
  \begin{tikzpicture}[baseline=(X.south)]
    \node[fill=#1,draw,inner sep=4pt] (X) {};
    \node[anchor=west,inner sep=0,align=left] at (.3,0) {\small #2};
  \end{tikzpicture}
  }%
For quick reference, we state the propositions in which the individual
results are proved and the techniques used in each case.
\begin{theorem}\label{thm: Tutte main result}%
  Let $(x,y)\in\Q^2$. %
  Under \cETH{},
  \begin{enumerate}[(i)]
    \item\label{thmi: Tutte general}
      \legend{blue!25!gray}{%
        \normalsize
      $\pp{Tutte}(x,y)$ cannot be computed in time $\exp(o(n))$\\
      if $(x-1)(y-1)\neq 1$ and $y\not\in\{0,\pm 1\}$,}
      \\[.3\baselineskip]
      {\upshape\footnotesize\mbox{}\hfill
      (
      Stretching and thickening
      ;
      Proposition~\ref{prop: individual points, multigraphs, nonzero q}
      in \S\ref{sec: thickening and interpolation}
      )}
    \item\label{thmi: Tutte linial}
      \legend{blue!25!gray}{%
        \normalsize
      $\pp{Tutte}^{0,1}(x,y)$ cannot be computed in time $\exp(o(n))$\\
      if $y=0$ and $x\not\in\{0,\pm 1\}$,}
      \\[.3\baselineskip]
      {\upshape\footnotesize\mbox{}\hfill
      (
      Linial's reduction
      ;
      Proposition~\ref{prop: Tutte linial} in
      Appendix~\ref{app: standard hardness results}
      )}
    \item\label{thmi: Tutte reliability}
      \legend{orange!80}{%
        \normalsize
      $\pp{Tutte}^{0,1}(x,y)$ cannot be computed in time $\exp(o( m/\log^2 m))$\\
      if $x=1$ and $y\neq 1$,} %
      \\[.3\baselineskip]
      {\upshape\footnotesize\mbox{}\hfill
      (
      Inflation with Wump graphs
      ;
      Proposition~\ref{prop: reliability}
      in \S\ref{sec: reliability}
      )}
    \item\label{thmi: Tutte simple}
      \legend{orange!50}{%
        \normalsize
      $\pp{Tutte}^{0,1}(x,y)$ cannot be computed in time $\exp(o( m/\log^3 m))$\\
      if $(x-1)(y-1)\not\in\{0,1\}$ and
      $(x,y)\not\in\{(-1,-1),(-1,0),(0,-1)\}$.} %
      \\[.3\baselineskip]
      {\upshape\footnotesize\mbox{}\hfill
      (
      Inflation with Theta graphs
      ;
      Proposition~\ref{prop: Tutte Theta}
      in \S\ref{sec: points simple}
      )}
  \end{enumerate}
\end{theorem}
\begin{figure}[tpb]
\[\vcenter{\hbox{
    \begin{tikzpicture}
      [line cap=round,line join=round,x=.4cm,y=.4cm,scale=1.4]
  \def\xmin{-2.5}; \def\ymin{-2.1};\def\xmax{3.3}; \def\ymax{3.0};

\begin{scope}
   % the plane
  \clip (\xmin,\ymin) rectangle (\xmax,\ymax);
   \fill[color=blue!25!gray] (\xmin,\ymin) rectangle (\xmax,\ymax);
   \draw[color=white, ultra thick] (\xmin,1)--(\xmax,1);
   \draw[color=orange!50, ultra thick] (\xmin,-1)--(\xmax,-1);

   \draw[fill,orange!80] (1,-1) circle (1pt);
   \draw[fill,orange!80] (1,0) circle (1pt);

   % Special points
   \draw[fill] (-1,-1) circle (1pt);
   \draw[fill] (-1,0) circle (1pt);
   \draw[fill] (0,-1) circle (1pt);
   \draw[fill] (0,0) circle (1pt);
   \draw[fill] (1,1) circle (1pt);

   % H1
  \draw[smooth,ultra thick,samples=100,domain=\xmin:.9] plot(\x,{1+1/(\x-1)});
  \draw[smooth,ultra thick,samples=100,domain=1.1:\xmax] plot(\x,{1+1/(\x-1)});
\end{scope}
  % axes
  \draw[thick,-latex] (\xmin,\ymin) -- (\xmax+.5,\ymin);
  \draw[shift={(\xmax,\ymin)}] node[below] {$x$};
  \foreach \x in {-1,0,1}
  \draw[shift={(\x,\ymin)}] (0pt,0pt) -- (0pt,-2pt)
                                   node[below] {$\x$};

  \draw[thick,-latex] (\xmin,\ymin) -- (\xmin,\ymax+.5);
  \draw[shift={(\xmin,\ymax)}] node[left] {$y$};
  \foreach \y in {-1,0,1}
      \draw[shift={(\xmin,\y)}] (0pt,0pt) -- (-2pt,0pt)
                                    node[left] {$\y$};
\end{tikzpicture}}}
\vcenter{\hbox{
\begin{tikzpicture}
      [line cap=round,line join=round,x=.4cm,y=.4cm,scale=1.4]
  \def\xmin{-2.5}; \def\ymin{-2.1};\def\xmax{3.3}; \def\ymax{3.0};

   % the plane
  \begin{scope}
  \clip (\xmin,\ymin) rectangle (\xmax,\ymax);

  \draw[fill, orange!50]  (\xmin,\ymin) rectangle (\xmax, \ymax);
  \draw[color=blue!25!gray, ultra thick] (\xmin,0)--(\xmax,0);
  \draw[color=white, ultra thick] (\xmin,1)--(\xmax,1);
  \draw[color=orange!80, ultra thick] (1,\ymin)--(1,\ymax);

   % H1
  \draw[smooth,ultra thick,samples=100,domain=\xmin:.9] plot(\x,{1+1/(\x-1)});
  \draw[smooth,ultra thick,samples=100,domain=1.1:\xmax] plot(\x,{1+1/(\x-1)});

   % Special points
   \draw[fill] (-1,-1) circle (1pt);
   \draw[fill] (-1,0) circle (1pt);
   \draw[fill] (0,-1) circle (1pt);
   \draw[fill] (0,0) circle (1pt);
   \draw[fill] (1,1) circle (1pt);
  \end{scope}
  % axes
  \draw[thick,-latex] (\xmin,\ymin) -- (\xmax+.5,\ymin);
  \draw[shift={(\xmax,\ymin)}] node[below] {$x$};
  \foreach \x in {-1,0,1}
  \draw[shift={(\x,\ymin)}] (0pt,0pt) -- (0pt,-2pt)
                                   node[below] {$\x$};

  \draw[thick,-latex] (\xmin,\ymin) -- (\xmin,\ymax+.5);
  \draw[shift={(\xmin,\ymax)}] node[left] {$y$};
  \foreach \y in {-1,0,1}
      \draw[shift={(\xmin,\y)}] (0pt,0pt) -- (-2pt,0pt)
                                    node[left] {$\y$};
\end{tikzpicture}
}}
\begin{array}{l}
  \legend{blue!25!gray}{no $\exp(o( n))$%\\
    }\\
  \legend{orange!80}{no $\exp(o(n/\log^2 n))$%\\
    }\\
  \legend{orange!50}{no $\exp(o(n/\log^3 n))$%\\
    }\\
  \legend{white}{no $n^{O(1)}$ %\\
%    \color{gray}\cite{JVW90}%
    }\\
  \legend{black}{$n^{O(1)}$ %\\
%    \color{gray}\cite{JVW90}%
    }
\end{array}
\]
\caption{\label{fig: Tutte}\small Exponential time complexity under
  \cETH{} of the Tutte plane for multigraphs (left) and simple graphs
  (right) in terms of $n$, the number of vertices. %
  The white line $y=1$ on the map is uncharted territory, and we only
  have the \#P-hardness. %
  The black hyperbola $(x-1)(y-1)=1$ and the four points close to the
  origin are in P.
  Everywhere else, in the shaded regions, we prove a lower bound
  exponential in $n$, or within a polylogarithmic factor of it.}
\end{figure}
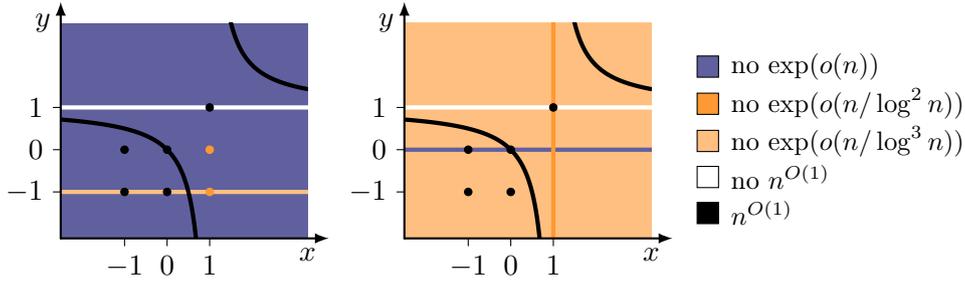%

Above, the results \iref{thmi: Tutte reliability} and \iref{thmi:
Tutte simple} are stated in terms of the parameter~$m$, the number of
edges of the given graph, but the same results also hold for the
parameter~$n$, the number of vertices, because $n\leq m$ in connected
graphs. %
The formulation with respect to~$m$ gives a stronger hardness result
under \cETH{} since~$m$ can potentially be much larger than~$n$.
This is in the same spirit as the sparsification lemma of
\citeN{IPZ01} and Theorem~\ref{thm: counting_sparsification_essence}.
Using this stronger formulation, Theorem~\ref{thm: Tutte main result}
can be used as a starting point for further hardness reductions under
\cETH{}.

In an attempt to prove Theorem~\ref{thm: Tutte main result}, we may
first turn to the literature, which contains a cornucopia of
constructions for proving hardness of the Tutte polynomial in various
models. %
In these arguments, a central role is played by graph transformations
called thickenings and stretches. %
A $k$-thickening replaces every edge by a bundle of~$k$ edges
\begin{tikzpicture}[scale=0.38]
  \node[draw,circle,fill=white,inner sep=1.2pt] (left)  at (0*1.4,0) {};
  \node[draw,circle,fill=white,inner sep=1.2pt] (right) at (1*1.4,0) {};
  \draw[bend left]  (left) edge (right);
  \draw             (left) edge (right);
  \draw[bend right] (left) edge (right);
\end{tikzpicture},
and a $k$-stretch replaces every edge by a path of $k$ edges
\begin{tikzpicture}[scale=0.38]
  \node[draw,circle,fill=white,inner sep=1.2pt] (left)  at (0*1.4,0) {};
  \node[draw,circle,fill=white,inner sep=1.2pt] (mid1)  at (1*1.4,0) {};
  \node[draw,circle,fill=white,inner sep=1.2pt] (mid2)  at (2*1.4,0) {};
  \node[draw,circle,fill=white,inner sep=1.2pt] (right) at (3*1.4,0) {};
  \draw[bend left]  (left) edge (mid1);
  \draw[bend left]  (mid1) edge (mid2);
  \draw[bend left]  (mid2) edge (right);
\end{tikzpicture}. %
This is used to ``move'' an evaluation from one point to another. %
For example, if $H$ is the $2$-stretch of~$G$ then
$T(H; 2,2)\sim T(G; 4,\frac{4}{3})$. %
Thus, every algorithm for $(2,2)$ works also at $(4,\frac{4}{3})$,
connecting the complexity of the two points. %
These reductions are very well-developed in the literature, and are
used in models that are immune to polynomial-size changes in the input
parameters, such as \cc{\#P}-hardness and approximation complexity. %
However, we cannot always afford such constructions in our setting,
otherwise our bounds would be of the form
$\exp(\Omega(n^{1/r}))$ for some constant $r$ depending on the blow-up
in the proof. %
In particular, the parameter $n$ is destroyed already by a $2$-stretch
in a non-sparse graph.

The proofs are in \S\ref{sec: hyperbolas}--\S\ref{sec: points simple}.
Where we can, we sample from established methods, carefully avoiding
or modifying those that are not parameter-preserving. %
At other times we require more subtle techniques, e.g., the
constructions in \S\ref{sec: points simple}, which use graph products
with graphs of polylogarithmic size instead of thickenings and
stretches. %
Like many recent papers, we use Sokal's multivariate version of the
Tutte polynomial, which vastly simplifies many of the technical
details.

\subsection*{Consequences}

The permanent and Tutte polynomial are equivalent to, or
generalisations of, various other graph problems, so our lower bounds
under \rETH{} and \cETH{} hold for these problems as well. %
In particular,  the
following graph polynomials (for example, as a list of their
coefficients) cannot be computed in time $\exp(o(m))$ for a given simple graph: the Ising partition function,
the $q$-state Potts partition function ($q\neq 0,1,2)$, the
reliability polynomial, the chromatic polynomial, and the flow
polynomial. %
Moreover, our results show that the following
counting problems on multigraphs cannot be solved in time $\exp(o(n))$: \#~perfect matchings, \# cycle
covers in digraphs, \# connected spanning subgraphs, all-terminal
graph reliability with given edge failure probability $p>0$, \#
nowhere-zero $k$-flows ($k\neq 0,\pm 1$), and \# acyclic orientations.

The lower bound for counting the number of perfect matchings holds
even in bipartite graphs, where an $O(1.414^n)$ algorithm is given by
Ryser's formula. %
Such algorithms are also known for general graphs~\cite{BH08}, the
current best bound is $O(1.619^n)$ \cite{K09}.

For simple graphs, we have $\exp(\Omega(m))$ lower bounds for
\#~perfect matchings and \#~cycle covers in digraphs.

\section{Counting Independent Sets}
\label{sec: indsets}

In this section, we establish Theorem~\ref{thm: 2sat}, the hardness of
counting independent sets and of \pp{\#$2$-Sat}.
For the proof, we make use of the randomized \ETH-hardness of the following
problem.
\begin{quote}\small
\begin{description}
  \item[Name]   $\pp{Unique $3$-Sat}$.
  \item[Input]  $3$-CNF formula $\varphi$ with $m$ clauses and at most
    one satisfying assignment.
  \item[Decide] Is $\varphi$ satisfiable?
\end{description}
\end{quote}
Calabro et al.~\cite{Calabro_isolation} prove an isolation lemma for $d$-CNF formulas
to show that solving this problem in subexponential time implies that
the (randomized) exponential time hypothesis fails.
\begin{theorem}[Corollary~2 of Calabro et al.~\cite{Calabro_isolation}]
  \mbox{}\\
  \rETH{} implies that \pp{Unique $3$-Sat} cannot be computed in time
  $\exp(o(m))$.
  \label{thm: Calabro_isolation}
\end{theorem}
We are now in the position to prove Theorem~\ref{thm: 2sat}.
\begin{theorem*}[Theorem \ref{thm: 2sat} (restated).]%
  Under \rETH{}, there is no randomized algorithm that computes the
  number of all independent sets in time
  $\exp(o(m))$, where $m$ is the number of edges.
  Under the same assumption, there is no randomized algorithm for
  \pp{\#$2$-Sat} that runs in time $\exp(o(m))$, where $m$ is the
  number of clauses.
\end{theorem*}
\begin{proof}%[of Theorem~\ref{thm: 2sat}]
  Let $\varphi$ be an instance of \pp{Unique $3$-Sat} with $m$
  clauses.
  We construct a graph~$G$ with $O(m)$ edges that has an odd number of
  independent sets if and only if $\varphi$ is satisfiable.
  For each variable~$x$, we introduce vertices $x$ and
  ${\overline x}$, and the edge~$(x {\overline x})$.
  This makes sure that any independent set of $G$ chooses at most one
  of $\{x,\overline{x}\}$, so we can interpret the independent set as
  a partial assignment to the variables of $\varphi$.
  For each clause $c=(\ell_1\vee\ell_2\vee\ell_3)$ of $\varphi$, we
  introduce a clique in~$G$ that consists of seven vertices
  $c_1,\dots,c_7$.
  These vertices correspond to the seven partial assignments that
  assign truth values to the literals $\ell_1$, $\ell_2$, and $\ell_3$
  ~in such a way that~$c$ is satisfied.
  Any independent set of~$G$ contains at most one~$c_i$ for each
  clause~$c$.
  To ensure that the independent set chooses the variables and partial
  assignments of the clauses consistently, we add an edge for
  every~$c_i$ and every variable~$x$ occurring in the clause~$c$:
  If the partial assignment that corresponds to~$c_i$ sets~$x$ to
  true, we add $(c_i \overline{x})$ to $G$;
  otherwise, we add $(c_i x)$ to $G$.
  To finalize the construction, we introduce guard vertices~$g_x$
  and~$g_c$ for every variable~$x$ and every clause~$c$, along with
  the edges $(g_x x)$, $(g_x \overline x)$, and $(g_c c_i)$ for
  $i=1,\dots,7$.

  We now prove that~$G$ has the required properties.
  First, any independent set contains at most~$n$ literal vertices and
  at most~$m$ clause vertices.
  \emph{Good} independent sets are those that contain exactly~$n$
  literal and~$m$ clause vertices (and no guard vertex).
  Good independent sets correspond to the satisfying assignments of
  $\varphi$ in a natural way.
  We now show that the number of bad independent sets is even.
  For this, let $S$ be a bad independent set, that is, $S$ is disjoint
  from $\{x,\overline{x}\}$ for some~$x$ or it is disjoint from
  $\{c_1,\dots,c_7\}$ for some clause~$c$.
  By construction, the neighbourhood of either~$g_x$ or~$g_c$ is
  disjoint from~$S$.
  Let~$g$ be the lexicographically first guard vertex whose
  neighbourhood is disjoint from~$S$.
  Both the sets $S\setminus\{g\}$ and $S\cup\{g\}$ are bad independent
  sets and~$S$ is one of these sets.
  Formally, we can therefore define a function that maps these sets
  onto each other.
  This function is a well-defined involution on the set of bad
  independent sets, and it does not have any fixed points.
  Therefore, the number of bad independent sets is even, and the
  parity of the number of independent sets of~$G$ is equal to the
  parity of the number of satisfying assignments of~$\varphi$.

  The above reduction shows that an $\exp(o(m))$-time algorithm for
  counting independent sets modulo $2$ implies an $\exp(o(m))$-time
  algorithm for \pp{Unique $3$-Sat}.
  By Theorem~\ref{thm: Calabro_isolation}, this implies that \rETH{}
  fails.

  To establish the hardness of \pp{\#$2$-Sat}, we reduce from counting
  independent sets.
  Let~$G$ be a graph.
  For each vertex~$v$, we introduce a variable~$v$, and each edge
  $(uv)$ becomes a clause $(\overline{u}\vee\overline{v})$.
  The satisfying assignments of the so constructed $2$-CNF formula are
  in one-to-one correspondence with the independent sets of~$G$.
\end{proof}

\section{The Permanent}
\label{sec: permanent}

This section contains the proof of Theorem~\ref{thm: perm}. %
With $[0,n]=\{0,1,\ldots, n\}$ we establish the reduction chain
$\pp{\#$3$-Sat} \red \pp{Perm}^{-1,0,1} \red \pp{Perm}^{[0,n]} \red
\pp{Perm}^{0,1}$ while taking care of the instance sizes.

\begin{theorem*}[Theorem~\ref{thm: perm} (restated).]\mbox{ }
  \begin{enumerate}[(i)]
    \item%\label{thmi: perm standard}%
      $\pp{Perm}^{-1,0,1}$ and $\pp{Perm}$
      cannot be computed in time $\exp(o(m))$ under \cETH.
    \item%\label{thmi: perm zeroone}%
      $\pp{Perm}^{0,1}$ cannot be computed in time $\exp(o(m/\log n))$ under \cETH.
    \item%\label{thmi: perm RETH}%
      $\pp{Perm}^{0,1}$ cannot be computed in time $\exp(o(m))$ under \rETH.
\end{enumerate}
\end{theorem*}
\begin{proof}%[of Theorem~\ref{thm: perm}]
  To establish \iref{thmi: perm standard}, we reduce $\pp{\#$3$-Sat}$
  in polynomial time to $\pp{Perm}^{-1,0,1}$ such that \mbox{$3$-CNF}
  formulas~$\varphi$ with~$m$ clauses are mapped to graphs~$G$
  with~$O(m)$ edges. %
  For technical reasons, we preprocess~$\varphi$ such that every
  variable~$x$ occurs equally often as~a positive literal and as a
  negative literal~$\bar x$ (e.g., by adding trivial clauses of the
  form $(x\vee \bar x \vee \bar x)$ to~$\varphi$). %
  We construct~$G$ with $O(m)$ edges and weights $w:E \to \{\pm 1\}$
  such that $\pp{\#Sat}(\varphi)$ can be derived from $\perm G$ in
  polynomial time. %
  For weighted graphs, the permanent is
  \begin{equation*}
    \perm G =\sum_{C \subseteq E} w(C)\,,\qquad
    \text{where $w(C)=\prod_{e\in C} w(e)$}\,.
  \end{equation*}
  The sum above is over all cycle covers $C$ of $G$, that is,
  subgraphs $(V,C)$ with an in- and outdegree of $1$ at every vertex.

  \begin{figure}[tbp]
    \[
    \vcenter{\hbox{
      %%%%%%%%%%%%%%%%%%%%%%%%%%%%%%%%%%%%%%%%%%%%%%%%%%%%%%% variable gadget
      \begin{tikzpicture}[node distance=1.5cm,auto]
        \tikzstyle{every edge}+=[dedge]

        \node [bullet]               (top) {};
        \node [bullet, below of=top] (bot) {};

        \draw [bend right=45] (top) edge node[left]  {$x$}      (bot);
        \draw [bend left=45]  (top) edge node[right] {$\bar x$} (bot);
        \draw                 (bot) edge (top);
      \end{tikzpicture}
}}
      %%%%%%%%%%%%%%%%%%%%%%%%%%%%%%%%%%%%%%%%%%%%%%%%%%%%%%% clause gadget
\qquad\vcenter{\hbox{
      \begin{tikzpicture}[auto]
        \tikzstyle{every edge}+=[dedge]
        \tikzstyle{every loop}=[]

        \node [bullet]             (bot) at (-90:1.0) {};
        \node [bullet]             (tol) at (150:1.0) {};
        \node [bullet]             (tor) at (30:1.0)  {};
        \node [bullet]             (mid) at (-90:.4)     {};

        \draw [bend left=50] (tol) edge node {$\bar{\ell_2}$} (tor);
        \draw [bend left=50] (tor) edge node {$\bar{\ell_3}$} (bot);
        \draw [bend left=50] (bot) edge node {$\bar{\ell_1}$} (tol);
        \draw [bend right=10](mid) edge (tor);
        \draw [bend left=10] (mid) edge (tol);
        \draw [bend right=10](tor) edge (tol);
        \draw [bend right=10] (tol) edge (tor);
        \draw [bend right=10](tor) edge (mid);
        \draw [bend left=10] (tol) edge (mid);
        \draw [bend right=20](mid) edge (bot);
        \draw [bend right=20](bot) edge (mid);
      \end{tikzpicture}
}}
      %%%%%%%%%%%%%%%%%%%%%%%%%%%%%%%%%%%%%%%%%%%%%%%%%%%%%%% equality gadget
\qquad\vcenter{\hbox{
    \begin{tikzpicture}[auto]
        \tikzstyle{every edge}+=[dedge]
        \tikzstyle{every loop}=[]

        \node [bullet,label=left:{$u$}] (ul)  at (135:1.5) {};
        \node [bullet]                  (uvl) at (180:0.5) {};
        \node [bullet,label=left:{$v$}] (vl)  at (205:1.5) {};

        \node [bullet,label=right:{$u\smash{'}$}] (ur)  at (45:1.5) {};
        \node [bullet]                            (uvr) at (0:0.5)   {};
        \node [bullet,label=right:{$v\smash{'}$}] (vr)  at (335:1.5) {};

        % paths from u to v and from u' to v'
        \draw [bend right=15] (ul)  edge (uvl);
        \draw [bend right=15] (uvl) edge (vl);
        \draw [bend left=15]  (ur)  edge (uvr);
        \draw [bend left=15]  (uvr) edge (vr);

        % left-right edges
        \draw [bend left=10]  (uvl) edge (uvr);
        \draw [bend left=10]  (uvr) edge (uvl);

        % node in the top...
        \node [bullet] (top) at (90:.8) {};
        % ... and its edges
        \draw [bend left=10] (uvl) edge (top);
        \draw [bend left=10] (top) edge (uvl);
        \draw [bend left=10] (uvr) edge (top);
        \draw [bend left=10] (top) edge (uvr);
        \draw [in=60,  out=120, loop] (top) edge node {$-1$} (top);
        \draw [in=-120,out=-60, loop] (uvl) edge (uvl);
        \draw [in=-120,out=-60, loop] (uvr) edge (uvr);
      \end{tikzpicture}}}
\]
\caption{\label{fig:perm-gadgets}%
  Left: A selector gadget for variable $x$. %
  Depending on which of the two cycles is chosen, we assume~$x$ to be
  set to true or false. %
  Middle: A clause gadget for the clause $\ell_1\vee\ell_2\vee\ell_3$.
  The gadget allows all possible configurations for the outer edges,
  except for the case that all three are chosen (which would
  correspond to $\ell_1=\ell_2=\ell_3=0$). %
  Right: An equality gadget that replaces two edges $uv$ and~$u'v'$.
  The top loop carries a weight of~$-1$. %
  It can be checked that the gadget contributes a weight of~$-1$ if
  all four outer edges are taken, $+2$ if none of them is taken,
  and~$0$ otherwise.  }
  \end{figure}
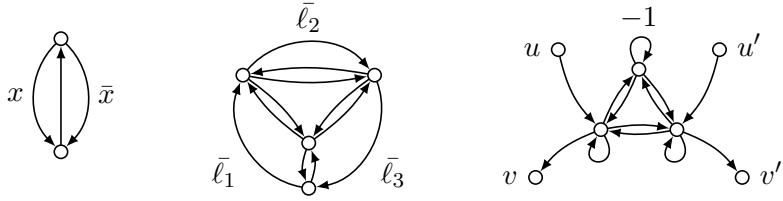
  In Figure~\ref{fig:perm-gadgets}, the gadgets of the construction are
  depicted. %
  For every variable~$x$ that occurs in $\varphi$, we add a
  \emph{selector gadget} to $G$. %
  For every clause $c=\ell_1 \vee \ell_2 \vee \ell_3$ of~$\varphi$, we
  add a \emph{clause gadget} to $G$. %
  Finally, we connect the edge labelled by a literal~$\ell$ in the
  selector gadget with all occurrences of $\ell$ in the clause
  gadgets, using \emph{equality gadgets}. %
  That is, we use a fresh copy of the equality gadget for each
  occurrence of a literal.
  For the first occurrence of the literal, we replace the corresponding
  edge in the selector gadget with a path of length two and identify
  this path with the path from $u$ to $v$ in the corresponding copy of
  the equality gadget.
  Furthermore, we replace the corresponding edge in the clause gadget
  with a path of length two and identify this path with the path from
  $u'$ to $v'$.
  For subsequent occurrences of the literal, we subdivide one of the
  edges on the corresponding path of the selector even further and use
  a new equality gadget as before.
  This concludes the construction of~$G$.

  The number of edges of the resulting graph $G$ is linear in the
  number of clauses. %
  The correctness of the reduction follows along the lines
  of~\cite{Papa} and~\cite{BD07}. %
  The satisfying assignments stand in bijection to cycle covers of
  weight~$(-1)^i 2^j$ where $i$ (resp.\ $j$) is the number of
  occurrences of literals set to false (resp.\ true) by the
  assignment, and all other cycle covers sum up to~$0$. %
  Since we preprocessed~$\varphi$ such that $i=j$ holds and $i$ is
  constant over all assignments, we obtain
  $\perm G=(-2)^i \cdot \pp{\#Sat}(\varphi)$.

  For the second part of \iref{thmi: perm standard}, we reduce
  $\pp{Perm}^{-1,0,1}$ in polynomial time to $\pp{Perm}^{[0,n]}$ by
  interpolation: %
  On input~$G$, we conceptually replace all occurrences of the
  weight $-1$ by a variable~$x$ and call this new graph~$G_x$. %
  We can assume that only loops have weight~$x$ in~$G_x$ because the
  output graph~$G$ from the previous reduction has weight~$-1$ only on
  loops. %
  Then $p(x)=\perm G_x$ is a polynomial of degree $d \leq n$.

  If we replace $x$ by a value $a\in [0,n]$, then $G_a$ is a weighted
  graph with as many edges as $G$. %
  As a consequence, we can use the oracle to compute $\perm G_a$ for
  $a=0,\dots,d$ and then interpolate, to get the coefficients of the
  polynomial~$p(x)$. %
  At last, we return the value $p(-1)=\perm G$.  This completes the
  reduction, which queries the oracle $d+1$ graphs that have at most
  $m$ edges each.

  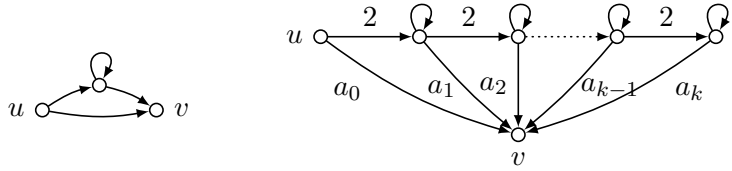
\begin{figure}[tpb]
    \[\vcenter{\hbox{
      \begin{tikzpicture}[auto]
        \tikzstyle{every loop}=[]
        \tikzstyle{every edge}+=[dedge]
        \node [bullet,label=left:{$u$}]  (u)   at (180:.75) {};
        \node [bullet]                   (top) at (90:0.33) {};
        \node [bullet,label=right:{$v$}] (v)   at (0:.75)   {};

        \draw [bend right=10] (u)    edge (v);
        \draw [bend left=10]  (u)    edge (top);
        \draw [bend left=10]  (top)  edge (v);
        \draw [in=60, out=120,loop]    (top)  edge (top);
      \end{tikzpicture}
}}
\qquad
    \vcenter{\hbox{
      \begin{tikzpicture}[node distance=1.3cm,auto]
        \tikzstyle{every loop}=[]
        \tikzstyle{every edge}+=[dedge]
        \node [bullet,label=left:{$u$}]                 (u)      {};
        \node [bullet,right of=u]      (two1)   {};
        \node [bullet,right of=two1]   (two2)   {};
        \node [bullet,right of=two2]   (twok-1) {};
        \node [below of=two2,bullet,label=below:{$v$}]  (v)      {};
        \node [bullet,right of=twok-1] (twok)   {};

        \draw [bend right=10](u)      edge node[near start,below left]  {$a_0$}    (v);
        \draw [bend right=5] (two1)   edge node[left]  {$a_1$}    (v);
        \draw [bend left=0]  (two2)   edge node[left]  {$a_2$}    (v);
        \draw [bend left= 5] (twok-1) edge node[right] {$a_{k-1}$} (v);
        \draw [bend left=10] (twok)   edge node[near start,below right] {$a_{k}$}  (v);

        \draw         (u)      edge node {$2$} (two1)  ;
        \draw         (two1)   edge node {$2$} (two2)  ;
        \draw [dotted](two2)   edge            (twok-1);
        \draw         (twok-1) edge node {$2$} (twok)  ;

        \draw [in=60, out=120,loop ] (two1)   edge (two1);
        \draw [in=60, out=120,loop ] (two2)   edge (two2);
        \draw [in=60, out=120,loop ] (twok-1) edge (twok-1);
        \draw [in=60, out=120,loop ] (twok)   edge (twok);
      \end{tikzpicture}}}\]
    \caption{\label{fig:perm-weight-a}%
    Left: This gadget simulates in unweighted graphs edges $uv$ of
    weight $2$. %
    Right: This gadget simulates edges $uv$ of weight
    $a=\sum_{i=0}^k a_i 2^i$ with $a_i\in\left\{0,1\right\}$.
    }
  \end{figure}
  For \iref{thmi: perm zeroone}, we have to get rid of weights larger
  than~$1$. %
  Let $G_a$ be one query of the last reduction. %
  Again we assume that $a\leq n$ and that weights $\neq 1$ are only
  allowed at loop edges. %
  We replace every edge of weight~$a$ by the gadget that is drawn in
  Figure~\ref{fig:perm-weight-a}, and call this new unweighted graph
  $G'$. %
  It can be checked easily that the gadget indeed simulates a weight
  of~$a$ (parallel paths correspond to addition, serial edges to
  multiplication), i.e., $\perm G'=\perm G_a$. %
  Unfortunately, the reduction increases the number of edges by a
  superconstant factor: The number of edges of~$G'$ is
  $m(G')\leq(m+n\log a)\leq O(m + n \log n)$. %
  But since $m(G')/\log m(G') \leq O(m)$, the reduction implies
  that~\iref{thmi: perm zeroone}.

  For \iref{thmi: perm RETH}, we assume that \rETH{} holds. %
  Theorem~\ref{thm: Calabro_isolation} gives that \pp{Unique $3$-Sat}
  cannot be computed in time $\exp(o(m))$.
  Now we apply the first reduction of~\iref{thmi: perm standard} to a
  formula~$\varphi$ which is promised to have at most one satisfying
  assignment.
  Then the number $\perm G=(-2)^i \cdot \pp{\#Sat}(\varphi)$ is either
  $0$ or $(-2)^i$.
  In $G$, we replace each edge of weight $-1$ by a gadget of
  weight~$2\equiv -1 \mod 3$ and similarly get that $(\perm G\bmod 3)$
  is $(0\bmod 3)=0$
  or $(4^i\bmod 3)=1$.
  Hence we can distinguish the case in which
  $\varphi$ is unsatisfiable from the case in which $\varphi$ has
  exactly one satisfying assignment.
\end{proof}

\section{Hyperbolas in the Tutte plane}
\label{sec: hyperbolas}

Consider a hyperbola in the Tutte plane described by $(x-1)(y-1)=q$, where $q$ is some fixed rational number.
Our first goal is to show that it is hard to compute the coefficients of the (univariate) restriction of the Tutte
polynomial to any such hyperbola.
It is useful to view the Tutte polynomial in the Fortuin--Kasteleyn
formulation \cite{fortuin1972random,Sokal2005}: %
\begin{equation} \label{eq: tutte FK-formulation}
  Z(G;q,w) = \sum_{A\subseteq E} q^{\comp(A)} w^{|A|}\,.
\end{equation}
Here, $\comp(A)$ is the number of connected components in the subgraph
$(V,A)$. %
The connection to the Tutte polynomial is given by
\begin{equation}\label{eq: sokal tutte}%
\begin{split}
 T(G;x,y)=(x-1)^{-\comp(E)}(y-1)^{-|V|} Z(G;q,w)\,,\\
 \text{where $q=(x-1)(y-1)$ and $w=y-1$\,,}
\end{split}
\end{equation}
see~\cite[eq.~(2.26)]{Sokal2005}. %

\subsection*{The Ising Hyperbola}

The Ising partition function is the Tutte polynomial from~\eqref{eq:
tutte FK-formulation} when~$q$ is fixed to~$2$.
We now show that computing the coefficients of this univariate
polynomial is hard under \cETH{}.
\begin{proposition}\label{prop: Ising hyperbola}
  If \cETH{} holds,
  the coefficients of the polynomial $w\mapsto Z(G;2,w)$ for
  a given simple graph $G$ cannot be computed in time $\exp(o(m))$.
\end{proposition}

\begin{proof}
  The reduction is from \pp{\#MaxCut} and well-known, see,
  \emph{e.g.}, \cite[Theorem~15]{JerrumSinclair}. %
  \begin{quote}\small
    \begin{description}
      \item[Name] $\pp{\#MaxCut}$
      \item[Input] Simple undirected graph $G$.
      \item[Output] The number of maximum cuts.
    \end{description}
  \end{quote}
  A maximum cut is a set $C\subseteq V(G)$ that maximizes the number
  $|E(C,\overline C)|$ of edges of~$G$ that cross the cut.
  By the Fortuin--Kasteleyn identity \cite[Theorem~2.3]{Sokal2005}, one
  can express $Z(G;2,w)$ for $G=(V,E)$ as
  \[
  \sum_{\sigma\colon V\mapsto \pm 1}
  \prod_{uv\in E} \bigl(1+w\cdot[\sigma(u)=\sigma(v)])\,.
  \]
  Here the Iverson bracket~$[P]$ is~$1$ if~$P$ is true and is~$0$
  if~$P$ is false.
  The sets $\sigma^{-1}(1)$ and $\sigma^{-1}(-1)$ define a cut in $G$,
  so we can write the above expression as
  \[
  \sum_{U\subseteq V}
  \prod_{\substack{uv\in E\\ [u\in U]=[v\in U]}} (1+w)\,
  =
  \sum_{C\subseteq V(G)} (1+w)^{m-|E(C,\overline C)|}\,,
  \]
  Now, the coefficient of $(1+w)^{m-c}$ in $Z(G;2,w)$ is the number of
  cuts in $G$ of size $c$. %
  In particular, after some interpolation, we can compute the number
  of maximum cuts in $G$ from the coefficients of $w\mapsto Z(G;2,w)$. %
  But as we observe in Appendix~\ref{app: standard hardness results},
  \pp{\#MaxCut} cannot be computed in time $\exp(o(m))$ under \cETH{}.
\end{proof}

\subsection*{The Multivariate Tutte Polynomial}

For other $q$, in particular non-integers, it is simpler to work with a
\emph{multivariate} formulation of the Tutte polynomial due to
\citeN{fortuin1972random}. %
We use the definition by~\citeN{Sokal2005}: %
Let $G=(V,E)$ be an undirected graph whose edge weights are given by a
function $\mathbf w\colon E\rightarrow\Q$. %
Then
\begin{equation}\label{eq: sokal def}%
  Z(G;q,\mathbf w)
  = \sum_{A\subseteq E}q^{k(A)}\prod_{e\in A}  \mathbf w(e)\,.
\end{equation}
If $\mathbf w$ is single-valued, in the sense that $\mathbf w(e)=w$
for all $e\in E$, we recover $Z(G;q,w)$. %

The conceptual strength of the multivariate perspective is that it
turns the Tutte polynomial's second variable~$y$, suitably
transformed, into an edge weight of the graph. %
In particular, the multivariate formulation allows the graph to have
different weights on different edges, which turns out to be a dramatic
technical simplification even when, as in the present work, we are
ultimately interested in the single-valued case.

Sokal's polynomial vanishes at $q=0$, so we sometimes use the
polynomial
\begin{equation*}
  Z_0(G;q,\mathbf w)
  =
  \sum_{A\subseteq E}q^{k(A)-k(E)}\prod_{e\in A} \mathbf w(e)\,,
\end{equation*}
which gives something non-trivial for $q=0$ and is otherwise a proxy
for $Z$:
\begin{equation} \label{eq: Z and Z_0}
  Z(G;q,\mathbf w) =q^{k(E)} Z_0(G;q,\mathbf w)\,.
\end{equation}

\subsection*{Three-terminal minimum cut}

For $q\not\in\{1,2\}$, we first establish that, with two different
edge weights, one of them negative, the multivariate Tutte polynomial
computes the number of $3$-terminal minimum cuts:
\begin{quote}\small
\begin{description}
  \item[Name] \pp{\#$3$-Terminal MinCut}
  \item[Input] Simple undirected graph $G=(V,E)$ with three
    distinguished vertices (``terminals'') $t_1,t_2,t_3 \in V$.
  \item[Output] The number of edge subsets $A\subseteq E$ of minimal
    size that separate $t_1$ from $t_2$, $t_2$ from $t_3$, and $t_3$
    from $t_1$.
\end{description}
\end{quote}
We establish the hardness of this problem under \cETH{} in
Appendix~\ref{app: standard hardness results}. %
The connection of this problem with the Tutte polynomial has been used
already by Goldberg and
Jerrum~\cite{GJ07FerromagneticIsing,GJ08Tutte}, with different
reductions, to prove hardness of approximation.

\medskip

The graphs we consider here are connected and have rather simple
weight functions. %
The edges are partitioned into two sets $E\dotcup T$ and, for some
fixed rational~$w$, the weight function is given by
\begin{equation}\label{eq: w def}
\mathbf w(e)=\begin{cases} -1, & \text{if $e\in T$},\\
 w, &\text{if $e\in E$}.
\end{cases}
\end{equation}

For such a graph, we have
\begin{equation}\label{eq: Z for G^T}
  Z_0(G;q,\mathbf w)=\sum_{A\subseteq E\cup T} q^{k(A)-1}w^{|A\cap E|
  } (-1)^{|A\cap T|}.
\end{equation}
For fixed $G$ and $q$, this is a polynomial in $w$ of degree at most
$m$.

\begin{lemma}\label{lem: three negative edges}
  Let $q$ be a rational number with $q\not\in\{1,2\}$. %
  The coefficients of the polynomial $w\mapsto\allowbreak
  Z_0(G;q, \mathbf w)$, with $\mathbf w$ as in \eqref{eq: w def}, for
  a given simple graph~$G$ cannot be computed in time $\exp(o(m))$ under
  \cETH{}.
  Moreover, this is true even if $|T|=3$.
\end{lemma}
\begin{proof}
  In Appendix~\ref{app: standard hardness results}, we argue that a
  standard reduction from \pp{\#MaxCut} already implies that the
  problem \pp{\#$3$-Terminal MinCut} cannot be computed in time $\exp(o(m))$
  under \cETH.
  We reduce this problem to the problem of evaluating the coefficients
  of~$Z_0$ at $q\not\in\{1,2\}$.
  Suppose $G'=(V,E,t_1,t_2,t_3)$ is an instance of \pp{\#$3$-Terminal
  MinCut} with $n=|V|$ and $m=|E|$. %
  We can assume that $G'$ is simple and connected. %
  We modify~$G'$ by adding a triangle between the terminals, obtaining
  the graph $G=(V,E\cup T)$ where $T=\{t_1t_2,t_2t_3,t_1t_3\}$; note
  that $n(G)=n$, $m(G)=m+3$, and $|T|=3$.

  We focus our attention on the family $\scr A$ of edge subsets
  $A\subseteq E$ for which $t_1$, $t_2$, and~$t_3$ each belong to a
  distinct component in the graph $(V,A)$. %
  In other words, $A$ belongs to~$\scr A$ if and only if $E-A$ is a
  3-terminal cut in~$G'$. %
  Then we can split the sum in \eqref{eq: Z for G^T} into
  \begin{equation}\label{eq: split}%
    Z_0(G; q,\mathbf w) =
    \sum_{B\subseteq T}
    \left(
      \sum_{A\in\scr A} q^{\comp (A\cup B)-1} w^{|A| } (-1)^{|B|}
      + \sum_{A\notin\scr A} q^{\comp (A\cup B)-1}w^{|A| } (-1)^{|B|}
    \right).
  \end{equation}

  We first show that the second term of \eqref{eq: split} vanishes. %
  Consider an edge subset $A\not\in \scr A$ and assume without loss of
  generality that it connects the terminals~$t_1$ and~$t_2$. %
  Consider $B\subseteq T$, and let $B'=B\oplus \{t_1t_2\}$, so that
  $B'$ is the same as $B$ except for $t_1t_2$. %
  Then the contributions of $A\cup B$ and $A\cup B'$ cancel: %
  First, $\comp(A\cup B)$ equals $\comp(A\cup B')$ because~$t_1$
  and~$t_2$ are connected through $A$ already, so the presence or
  absence of the edge $t_1t_2$ makes no difference. %
  Second, $(-1)^{|B|}$ equals $-(-1)^{|B'|}$.

  We proceed to simplify the first term of \eqref{eq: split}. %
  The edges in $B$ only ever connect vertices in $T$, and for
  $A\in\scr A$, each of these lies in a separate component of $(V,A)$,
  so
  \begin{equation*}
    \comp(A\cup B)=
    \begin{cases}
      \comp(A)-|B|\,,  & \text{if $|B|=0,1,2$},\\
      \comp(A)-2 \,,   &\text{if $|B|=3$}.
    \end{cases}
  \end{equation*}
  Calculating the contribution of $B$ for each size $|B|$, we arrive at
  \begin{equation*}
    \sum_{B\subseteq T} \sum_{A\in\scr A} q^{\comp (A\cup B)-1}w^{|A|}(-1)^{|B|} =
    \sum_{A\in\scr A} q^{\comp (A)-1} (q^0 -3q^{-1} + 3q^{-2} -q^{-2})w^{|A|}\,,
  \end{equation*}
  and after some simplifications we can write~\eqref{eq: split} as
  \begin{equation}
    Z_0(G;q,\mathbf w) = Q\cdot \sum_{A\in\scr A}
    q^{\comp (A)-3}w^{|A|}\,,\quad
    \text{where $Q=(q-1)(q-2)\,$.}
  \end{equation}
  Note that, by assumption on $q$, we have $Q\neq 0$.

  Let us write $\sum_{i=0}^m d_i w^i = Q^{-1} Z_0(G;q,\mathbf w)$,
  i.e., $d_i$ is the coefficient of the monomial~$w^i$ in the sum
  above. %
  More specifically,
  \begin{equation*}
    Q\cdot d_i = \sum_{A \in \scr A\colon |A|=i} q^{\comp(A) - 3}\,.
  \end{equation*}

  The edge subsets $A\in\scr A$ are exactly the complements of the
  $3$-terminal cuts in~$G'$. %
  Now consider the family $\scr C$ of \emph{minimal} $3$-terminal
  cuts, all of size $c$. %
  The sets $E-A$ in~$\scr C$ are exactly the sets $A$ of size $m-c$
  in~$\scr A$, and by minimality, $\comp(A) = 3$. %
  Thus, \[ Q\cdot d_{m-c} = \sum_{A\in \scr A\colon |A|=m-c} q^{3-3}
  =|\scr C|. \]

  Thus, if we could compute the coefficients $d_0,\ldots, d_m$
  of $w\mapsto Q^{-1} Z_0(G;q,\mathbf w)$, then we could determine the
  smallest $c$ so that $d_{m-c} \neq 0$ and return
  $d_{m-c}=|\scr C|/Q$, the number of $3$-terminal mincuts.
\end{proof}

\subsection*{General Hyperbolas}

We use Lemma~\ref{lem: three negative edges} to show that the
coefficients of the univariate Tutte polynomial from~\eqref{eq: tutte
FK-formulation} are hard to compute for any fixed $q\not\in\{1,2\}$. %
For this, we need to get rid of negative weights and reduce to a
single-valued weight function. %
\citeN{GJ08Tutte} achieve this using stretching and thickening, which
we want to avoid. %
Since the number of edges with a negative weight is small
(in fact, $3$), we can use another tool: deletion--contraction. %

\medskip
\newcommand{\etriangle}{\begin{tikzpicture}
	\path (90+0*120:0.15) node[vertex] (v0) {};
	\path (90+1*120:0.15) node[vertex] (v1) {};
	\path (90+2*120:0.15) node[vertex] (v2) {};
	\draw (v0) -- (v1) -- (v2) -- (v0);
	\draw[very thick] (v2) -- (v0);
\end{tikzpicture}}
A \emph{deletion--contraction} identity expresses a function of the graph
$G$ in terms of two graphs $G-e$ and $G / e$, where $G-e$ arises from $G$
by \emph{deleting} the edge $e$
($\etriangle \mapsto \begin{tikzpicture}
	\path (90+0*120:0.15) node[vertex] (v0) {};
	\path (90+1*120:0.15) node[vertex] (v1) {};
	\path (90+2*120:0.15) node[vertex] (v2) {};
	\draw (v0) -- (v1) -- (v2);
\end{tikzpicture}$)
and $G/e$ arises from $G$ by
\emph{contracting} the edge $e$
($\etriangle \mapsto \begin{tikzpicture}
	\node[vertex] (left)  at (0,0) {};
	\node[vertex] (right) at (0.4,0) {};
	\draw[bend left]  (left) edge (right);
	\draw[bend right] (left) edge (right);
	\end{tikzpicture}$)
that is, deleting it and identifying its endpoints (so any remaining
edges between these two endpoints become loops).

It is known~\cite[eq.~(4.6)]{Sokal2005} that
\begin{equation*}\label{eq: delcon}%
  Z(G;q,\mathbf w)= Z(G- e; q,\mathbf w) + \mathbf w(e)Z(G/e;q,\mathbf w).
\end{equation*}
An edge~$e$ is a \emph{bridge} of~$G$ if deleting~$e$ from~$G$
increases the number of connected components.
The above gives a deletion--contraction identity for $Z_0$ as well:
\begin{equation} \label{eq: delcontr}
	\ZNUL{G}{q}{\mathbf w} =
	\begin{cases}
		\;q \ZNUL{G- e}{q}{\mathbf w} +
		\mathbf w(e) \ZNUL{G/ e}{q}{\mathbf w}
		& \text{if $e$ is a bridge},\\
		  \;\;\,
		\ZNUL{G- e}{q}{\mathbf w} +
		\mathbf w(e) \ZNUL{G/ e}{q}{\mathbf w}
		& \text{otherwise}.
	\end{cases}
\end{equation}
\begin{proposition}\label{prop: hyp}
  Let $q$ be a rational number with $q\notin\{ 1,2\}$. %
  The coefficients of the polynomial $v\mapsto Z_0(G;q,v)$ for a given
  simple graph $G$ cannot be computed in time $\exp(o(m))$ under
  \cETH{}.
\end{proposition}
By \eqref{eq: Z and Z_0}, this proposition also holds for $Z$ instead
of $Z_0$ when $q\not\in\{0,1,2\}$.
\begin{proof}
  Let $G=(V,E)$ be a graph as in the previous lemma, with three edges
  $T=\{e_1,e_2,e_3\}$ of weight $-1$. %
  The given reduction actually uses the restriction that
  $G'=(V,E\setminus T)$ is connected, so we can assume that this is
  the case. %
  Thus, none of the $T$-edges is a bridge, so three applications of
  \eqref{eq: delcontr} to delete and contract these edges, gives
  \begin{equation}\label{eq: delcon expansion}
    Z_0(G;q,\mathbf w)=
    \sum_{C\subseteq \{1,2,3\}} (-1)^{|C|} Z_0(G_C;q,\mathbf w),
  \end{equation}
  where for each $C\subseteq \{1,2,3\}$, the graph $G_C$ is
  constructed from $G$ by removing $e_1,e_2,e_3$ as follows: %
  If $i\in C$ then $e_i$ is contracted, otherwise it is deleted. %
  In any case, the edges of~$T$ have disappeared and remaining edges
  of $G_C$ are in one-to-one correspondence with the edges in $E$;
  especially, they all have the same weight $w$, so $Z_0(G_C;q,\mathbf
  w)=Z_0(G_C;q,w)$.

  The resulting $G_C$ are not necessarily simple, because the
  contracted edges from~$T$ may have been part of a triangle and may
  have produced a loop. %
  (In fact, investigating the details of the previous lemma, we can
  see that this is indeed the case.) %
  Thus we construct the simple graph~$G_C'$ from~$G_C$ by subdividing
  every edge into a 3-path. %
  This operation, known as a 3-stretch, is known to largely preserve
  the value of $Z$ and $Z_0$ (see \cite{Sokal2004} for the former and
  \cite{GJ08Tutte} for the latter). In particular,
  \[
    Z_0(G_C;q,w) = f(q,w')^m \cdot Z_0(G_C'; q, w')\,,
  \]
  where for $q\neq 0$
  \[
 	1+\frac{q}{w}= \Bigl(1+\frac{q}{w'}\Bigr)^3
  	\quad \text{and} \quad f(q,w')=q^{-1}\cdot ((q+w')^3-{w'}^3)  \,,
  \]
  and for $q=0$
  \[
 	w = w'/3 \quad \text{and} \quad f(q,w')=1/(3w'^2)  \,.
  \]
  In summary, to compute the coefficients of the polynomial
  $w\mapsto Z_0(G;q,\mathbf w)$, we need to compute the 8 polynomials
  $v\mapsto Z_0(G_C; q,v)$, one for each $G_C$. We use the above
  equation and the assumed oracle for simple graphs to do this. %
  We note that every $G_C'$ is simple and has at most $n+2m$ vertices
  and at most $3m$ edges.
\end{proof}

\section{Individual Points for Multigraphs}
\label{sec: thickening and interpolation}
If we allow graphs to have multiple edges, we can use thickening and
interpolation, one of the original strategies of \citeN{JVW90}, for
relocating the hardness result for hyperbolas from
Proposition~\ref{prop: Ising hyperbola} and Proposition~\ref{prop: hyp} to
individual points in the Tutte plane. %
For most points, this gives us tight bounds in terms of~$n$, the
number of vertices, but not for points with $y\in\{0,\pm 1\}$, where
thickening fails completely.

We recall the thickening identities for the Tutte polynomial. %
The \emph{$k$-thickening of~$G$} is the graph~$G_k$ in which all edges
have been replaced by~$k$ parallel edges. %
One can show~\cite[(4.21)]{Sokal2005} that, with
$w_k = (1+w)^{k} - 1$,
\begin{equation}\label{multitutte thickening}%
  Z(G;q,w_{k}) = Z(G_{k};q,w)\,.
\end{equation}
It is easy to transfer this result to the Tutte polynomial~$T$
using~\eqref{eq: sokal tutte}, yielding special cases of Brylawski's
well-known graph transformation rules. %

We use interpolation and obtain
Theorem~\ref{thm: Tutte main result}\iref{thmi: Tutte general}
for $y\neq 0$ from the following.
\begin{proposition}\label{prop: individual points, multigraphs, nonzero q}
  Let $(q,w)\in\Q^2$ with $w\not\in \{0,-1,-2\}$ and $q\neq 1$.

  $Z(G;q,w)$ for a given graph $G$ (not necessarily simple)
  cannot be computed in time $\exp(o(n))$ under \cETH{}.
\end{proposition}

\begin{proof}
  We observe that the values $w_k=(1+w)^k-1$ are all distinct for
  $k=0,1,\dots,m$.
  Thus, the $k$-thickenings $G_k$ of $G$ give rise to $m+1$ different
  weight shifts, the evaluations of which, $Z(G;q,w_k)$, can be
  obtained from $Z(G_k;q,w)$ using~\eqref{multitutte thickening}.
  Thus, with oracle access to $G'\mapsto Z(G';q,w)$, we can compute
  the coefficients of the polynomial $v\mapsto Z(G;q,v)$ in polynomial
  time for any given~$G$.
  By Proposition~\ref{prop: Ising hyperbola} and Proposition~\ref{prop: hyp},
  this cannot be done in time $\exp(o(n))$ under \cETH.
  Since the number of vertices is $n$ in each $G_k$,
  computing $G'\mapsto Z(G';q,w)$ cannot be done in time
  $\exp(o(n))$ under \cETH.
\end{proof}
The proof of
Theorem~\ref{thm: Tutte main result}\iref{thmi: Tutte linial}
uses Linial's well-known reduction for the chromatic
polynomial~\cite{L86} , and is deferred to
Proposition~\ref{prop: Tutte linial} in
Appendix~\ref{app: standard hardness results}.

\section{Individual Points for Simple Graphs}
\label{sec: points simple}

In this section we show that most points $(x,y)$ of the Tutte plane
are as hard as the entire hyperbola on which they lie, even for
sparse, simple graphs. %
The drawback of our method is that we lose a polylogarithmic factor in
the exponent of the lower bound. %
The results are particularly interesting for the points on the line
$y=-1$, for which we know no other good exponential lower bounds under
\cETH{}, even in more general graph classes. %
We remark that the points $(-1,-1)$, $(0,-1)$, and $(\frac{1}{2},-1)$
on this line are known to admit a polynomial-time algorithm, and
indeed our hardness result does not apply here.  %

\subsection*{Graph inflations}
We use the graph theoretic version of Brylawski's tensor product for
matroids \cite{brylawski1982tutte}.
We found the following terminology more intuitive in our setting.
\begin{definition}[Graph inflation]
  Let $H$ be an undirected graph with two dis\-tin\-gui\-shed vertices called {\em terminals}. For any undirected graph
  $G=(V,E)$, an {\em $H$-inflation} of $G$, denoted \inflate{G}{H},
  is obtained by replacing every edge $xy\in E$ by (a fresh copy of)
  $H$, identifying $x$ with one of the terminals of~$H$ and~$y$ with
  the other.
\end{definition}
If $H$ is not symmetric with respect to its two terminals, then the
graph~\inflate{G}{H} need not be unique since there are in general two
non-isomorphic ways two replace an edge~$xy$ by~$H$.
For us this difference does not matter since the resulting Tutte
polynomials turn out to be the same; in fact, in any graph one can
remove a maximal biconnected component and reinsert it in the other
direction without changing the Tutte polynomial, an operation that is
called the \emph{Whitney twist} \cite{whitney19332}. %
Thus we choose \inflate{G}{H} arbitrarily among the graphs that
satisfy the condition in the definition above.
Graph inflation is not commutative and Sokal uses the
notation~$\vec G^H$.

If $H$ is a simple path of $k$ edges, \inflate{G}{H} gives the usual
$k$-stretch of~$G$, and a bundle of $k$ parallel edges results in
a $k$-thickening.
What makes graph inflations so useful in the study of Tutte
polynomials is that the Tutte polynomial of \inflate{G}{H} can be
expressed in terms of the Tutte polynomials of~$G$ and~$H$, so that
$Z(G\otimes H;q, w) \sim Z(G;q,w')$ for some ``shifted'' weight $w'$.

For fixed rational points $(q,w)$, we want to use interpolation to
prove the hardness of computing $Z(G;q,w)$ for a given graph $G$.
The basic idea is to find a suitable class of graphs $\{H_i\}$, such
that we can compute the coefficients of the univariate polynomial
$v\mapsto Z(G;q,v)$ for given~$G$ and~$q$ by interpolation from
sufficiently many evaluations of
$Z(G;q,w_i) \sim Z(G\otimes H_i;q,w)$.
For this, we need that the number of different weight shifts $\{w_i\}$
provided by the graph class $\{H_i\}$ is at least $|E(G)|+1$, one more
than the degree of the polynomial.

\subsection*{Generalised Theta Graphs}
\label{sec: gen theta}

For a set $S=\{s_1,\ldots,s_k\}$ of positive integers, the
\emph{generalised Theta graph} $\Theta_S$ consists of two vertices $x$
and $y$ joined by $k$ internally disjoint paths of $s_1,\ldots,s_k$
edges, respectively. %
For example,
\[
\Theta_{\{2,3,5\}}\quad\text{is }
\vcenter{\hbox{
  \tikzstyle{vertex}=[circle,fill=white,minimum size=3pt,draw, inner sep=0pt]
  \begin{tikzpicture}
    \node (x) [vertex,label=left:{$x$}] {};
    \node (y) [vertex,label=right:{$y$},right of=x] {};
    \draw
      (x) to node[pos=.33,vertex]{} node[pos=.66,vertex]{} (y)
      (x) to[bend right=45] node[pos=.5,vertex]{}          (y)
      (x) to[bend left=45]
        node[pos=.2,vertex]{}  node[pos=.4,vertex]{}
        node[pos=.6,vertex]{}  node[pos=.8,vertex]{}       (y);
  \end{tikzpicture}}
}\,.\]
For such graphs $\Theta_S$, we study the behaviour of the {\em Theta
inflation} $G\otimes \Theta_S$.

The Tutte polynomial of Theta graphs has already been studied by Sokal
in the context of complex roots of the chromatic polynomial. %
The necessary formulas for $Z(G\otimes \Theta_{S})$ can be derived
from \cite[prop 2.2, prop 2.3]{Sokal2004}. %
We present them here for the special case where all edge weights are
the same.

\begin{lemma}[Sokal]\label{lemma: theta shift}%
  Let $q$ and $w$ be rational numbers with $w\neq 0$
  and $q\not\in\{0,-2w\}$. %
  Then, for all graphs $G$ and finite sets $S$ of positive integers,
  \begin{equation}\label{eq: Z tensor}
    Z(G\otimes \Theta_S;q,w)
    =
    q^{-|E|\cdot|S|}
    \cdot \prod_{s\in S}\big( (q+w)^s - w^s \big)^{|E|}
    \cdot Z(G;q,w_S)\,,
  \end{equation}
  where
  \begin{equation}\label{eq: w_s definition}
    w_S
    = - 1 +
    \prod_{s\in S} \left( 1 + \frac{q}{(1+q/w)^s - 1}\right)\,.
  \end{equation}
\end{lemma}
This lemma can be derived from Sokal's series and parallel reduction
rules for $Z$ using a straightforward calculation.
Since all edge weights are the same, the result can also be
established from the classical Tutte polynomial via the series and
parallel reduction rules in \cite{JVW90}, but the calculation would be
slightly more laborious.

We now show that the class of Theta graphs provides a rich enough
spectrum of weight shifts to allow for interpolation.
In the following lemma, we use the definition of~$w_S$ from~\eqref{eq:
w_s definition}.
\begin{lemma}\label{lem: w_s}
  Let $q$ and $w$ be rational numbers with $w\neq 0$ and
  $q\not\in\{0,1,-w,-2w\}$. %
  For all integers $m\geq 1$, there exist sets $S_0,\dots,S_m$ of
  positive integers such that
  \begin{enumerate}[(i)]
    \item\label{conclusion w_s total size}
      $\sum_{s\in S_i} s \leq O(\log^3 m)$ for all $i$, and
    \item\label{conclusion w_s distinct}
      $w_{S_i} \neq w_{S_j}$ for all $i\neq j$.
  \end{enumerate}
  Furthermore, the sets $S_i$ can be computed in time polynomial
  in~$m$.
\end{lemma}

\begin{proof}
  Let $b=|1+q/w|$ and $f(s)=1+q/(b^s-1)$ for $s>0$. %
  Our choice of parameters ensures that $b>0$ and $b\neq 1$, so $f$ is
  a well-defined, continuous, and strictly monotone function from
  $\R^+\to\R$. %
  Furthermore, $w_S=-1+\prod_{s\in S} f(s)$ for all finite sets~$S$ of
  positive even integers. %
  Now let~$s_0\geq 2$ be an even integer such that~$f(s)$ is nonzero
  and has the same sign as~$f(s_0)$ for all $s\geq s_0$.
  For $i=0,\dots,m$, let $b_\ell\cdots b_0$ denote the binary
  expansion of $i$ where $\ell=\lfloor \log m\rfloor$. %
  Let $\Delta>6$ be a gap parameter that is a large and even integer
  chosen later, but only depends on $q$ and $w$.
  We define
  \begin{equation*}
    S_i =
    \Big\{\, s_0 + \Delta\lceil\log m\rceil \cdot (2j+b_j)
      \;\colon\;
    0\leq j \leq \ell\,\Big\}\,.
  \end{equation*}

  The salient feature of this construction is that all sets $S_i$ are
  different, of equal small cardinality, contain only positive even
  integers, and are from a range where~$f$ does not change sign. %
  Most important for our analysis is that the elements of the~$S_i$
  are spaced apart significantly, i.e.,
  \begin{equation}\tag{P}\label{salient: spaced apart}%
    \text{for $i,j$ and any $s\in S_i$ and $t\in S_j$,
    either $s=t$ or $|s-t|\geq \Delta\log m$.}
  \end{equation}
  From $|S_i|=\lfloor\log m\rfloor +1$ and the fact that all numbers
  in the sets are bounded by $O(\log^2 m)$, we immediately
  get~\iref{conclusion w_s total size}.

  \medskip To establish~\iref{conclusion w_s distinct}, let $0\leq i
  < j \leq m$. %
  We want to show that $w_{S_i} \neq w_{S_j}$. %
  Let us define $S=S_i\setminus S_j$ and $T=S_j\setminus S_i$. %
  From~\eqref{eq: w_s definition}, we see by multiplying with $(w_{S_i
    \cap S_j}+1)$ on both sides that $w_S+1 = w_T+1$ is equivalent to
  $w_{S_i} = w_{S_j}$ since $w_{S_i \cap S_j}\neq -1$.

  It remains to show that $\prod_{s\in S} f(s) \neq \prod_{t\in T}
  f(t)$. %
  Equivalently,
  \begin{equation}\label{eq: diff}
  \prod_{s\in S} \big(b^s+q-1\big) \prod_{t\in T} \big(b^t-1\big)
  -
  \prod_{t\in T} \big(b^t+q-1\big) \prod_{s\in S} \big(b^s-1\big)
  \neq 0
  \end{equation}

  We will multiply out the products in \eqref{eq: diff}. %
  Using the notation $\|X\|=\sum_{x\in X} x$, we rewrite
  \begin{equation*}
    \prod_{s\in S} \big(b^s+q-1\big) \prod_{t\in T} \big(b^t-1\big)
    =
    \sum_{X\subseteq S\cup T}
      (-1)^{|T\setminus X|} (q-1)^{|S\setminus X|} b^{\|X\|}\,.
  \end{equation*}
  Here we use the convention that for $X\subseteq S\cup T$, the term
  $b^s$ is taken in the first factor if $s\in X\cap S$, and $b^t$ is
  taken in the second factor if $t\in X\cap T$. %
  Doing this for both terms of~\eqref{eq: diff} and collecting terms
  we arrive at the equivalent claim
  \begin{equation}\label{eq: g claim}
    \sum_{X\subseteq S\cup T} g(X) \neq 0\,,
  \end{equation}
  where
  \begin{equation}\label{eq: g}
    g(X)=
    \left(
    (-1)^{|T\setminus X|} (q-1)^{|S\setminus X|}
    -
    (-1)^{|S\setminus X|} (q-1)^{|T\setminus X|}
    \right) \cdot b^{\|X\|}\,.
  \end{equation}
  Let $s_1$ be the smallest element of $S\cup T$ and without loss of
  generality assume that $s_1\in S$ (otherwise exchange $S$ and
  $T$). %
  Now from~\eqref{eq: g} and $|S|=|T|$, it follows that
  \begin{align*}
    g\big(S\cup T\big) =g(\emptyset) &=0\\
    g\big((S\cup T)\setminus\{s_1\}\big) &= q \cdot b^{\|S\cup T\|-s_1}\\
    g\big(\{s_1\}\big) &= (-q)\cdot(1-q)^{|S|-1} \cdot b^{s_1}\,.
  \end{align*}
  Since $q\neq0$, the largest exponent of~$b$ with nonzero coefficient
  in~\eqref{eq: g} is ${\|S\cup T\|-s_1}$ and all other exponents are
  at least $\Delta \log m$ smaller than that. %
  Similarly since $q\not\in\{0,1\}$, the smallest exponent of~$b$ with
  nonzero coefficient is $s_1$ and all other exponents are at least
  $\Delta \log m$ larger. %

  We let $X_0$ be the index in \eqref{eq: g claim} that maximizes the
  value $|g(X_0)|$.
  By the above considerations, we have $X_0=S\cup T\setminus\{s_1\}$
  for $b>1$ and $X_0=\{s_1\}$ for $b<1$.
  The total contribution of the remaining terms is $h=\sum_{X\neq X_0}
  g(X)$. %
  We prove~\eqref{eq: g claim} by showing $|h|<|g(X_0)|$. %
  From the triangle inequality and the fact that $S\cup T$ has at
  most~$4m^2$ subsets~$X$, we get
  \begin{equation*}
    |h|  \leq 4m^2 \cdot \max_{X\neq X_0}|g(X)|\\
         \leq 4m^2 \cdot 2|q-1|^{1+\log m}\cdot b^{\|X_0\| \pm \Delta\log m}
  \end{equation*}
  where the sign in $\pm\Delta\log m$ depends on whether~$b$ is larger
  or smaller than~$1$. %
  If ${b>1}$, the sign is negative. %
  In this case, notice that $\Delta=\Delta(q,w)$ can be chosen such that
  $4m^2\cdot 2|q-1|^{1+\log m} < |q|\cdot b^{\Delta \log m}$ for all
  $m\geq 2$.  %
  If $b<1$, we can similarly choose $\Delta$ as to satisfy $4m^2\cdot
  2|q-1|^{1+\log m} < |q|\cdot|1-q|^{|S|-1}\cdot b^{-\Delta \log
    m}$. %
  Thus, in both cases we have $|h|<|g(X_0)|$,
  which establishes~\iref{conclusion w_s distinct}.
\end{proof}

\subsection*{Points on the Hyperbolas}

The following proposition establishes
Theorem~\ref{thm: Tutte main result}\iref{thmi: Tutte simple},
which states that $Z$ is hard to evaluate at most points $(q,w)$ with
$q\not\in\{0,1\}$.
\begin{proposition}\label{prop: Tutte Theta}
  Let $(q,w)\in\Q^2\setminus\{(4,-2),(2,-1),(2,-2)\}$
  with $q\notin\{0,1\}$ and \mbox{$w\neq 0$}.
  If \cETH{} holds, then
  $Z(G;q,w)$ for a given simple graph $G$ cannot be computed in time
  $\exp(o(m/\log^3 m))$.%
\end{proposition}
By \eqref{eq: sokal tutte}, the points
$(4,-2)$, $(2,-1)$, and $(2,-2)$ in the $(q,w)$-plane
correspond to the polynomial-time computable
points $(-1,-1)$, $(-1,0)$, and $(0,-1)$ in the $(x,y)$-plane.
\begin{proof}
  We reduce from the problem of computing the coefficients of the
  polynomial $v\mapsto Z(G;q,v)$, which cannot be done in time
  $\exp(o(m))$ for $q\not\in\{0,1\}$ by Proposition~\ref{prop: Ising
    hyperbola} and Proposition~\ref{prop: hyp} (assuming \cETH{}).  We
  interpolate as in the proof of Proposition~\ref{prop: individual
    points, multigraphs, nonzero q}, but instead of thickenings we use
  Theta inflations to keep the number of edges relatively small.

  First we consider the degenerate case in which $q=-w$ or $q=-2w$.
  For a positive integer constant~$k$,
  let $G'$ be the $k$-thickening of~$G$.
  This transformation shifts the weight to $w'$ with
  \begin{equation*}
    w' = (1+w)^k-1\,,
  \end{equation*}
  which allows us to compute $Z(G;q,w')$ from $Z(G';q,w)$
  using~\eqref{multitutte thickening}.
  In the case $q=-w$, we have $1+w=1-q$, which cannot be~$1$ or~$0$,
  but which can also not be~$-1$ since then $(q,w)=(2,-2)$.
  Similarly, in the case $q=-2w$, we have $1+w=1-q/2$,
  which cannot be~$1$.
  It can also not be~$0$ since then $(q,w)=(2,-1)$,
  neither can it be~$-1$ since then $(q,w)=(4,-2)$.
  Thus, in any case, $(1+w)\not\in\{0,\pm 1\}$.
  This means that we can choose~$k$ large enough so that
  $q\not\in\{-w',-2w'\}$.
  This remains true if we let $G''$ be the $2$-stretch to~$G'$, which
  shifts the weight to $w''$ with
  \begin{equation*}
    1+\frac{q}{w''} = \left(1+\frac{q}{w'}\right)^2\,,
  \end{equation*}
  so that $Z(G;q,w'')$ can be computed from $Z(G'';q,w)$ (see
  \cite{Sokal2004}).
  We choose~$k$ so that $q\not\in\{-w'',-2w''\}$.
  The graph $G''$ after this transformation is simple and the number
  of edges is only increased by a constant factor of~$2k$.

  By the above, we can assume w.l.o.g.\ that $q\not\in\{-w,-2w\}$.
  We observe that the conditions $w\neq 0$
  and $q\not\in\{0,1,-w,-2w\}$ of Lemma~\ref{lem: w_s} now hold, and
  thus we can compute $m+1$ sets $S_0,S_1,\dots, S_m$ with all
  distinct weight shifts $w_0,\dots,w_m$ under Theta inflations.

  For a given graph $G$, let $G_i=G\otimes \Theta_{S_i}$.
  Using Lemma~\ref{lemma: theta shift}, we can compute the values
  $Z(G;q;w_i)$ from $Z(G_i;q,w)$.
  Moreover, as is clear from~\eqref{eq: sokal def}, the function
  $v\mapsto Z(G;q,v)$ is a polynomial of degree at most $m$, so we can
  use interpolation to recover its coefficients.  %
  We remark that the $G_i$ are simple graphs with at most
  $O(m\log^3 m)$ edges, so the claim follows. %
\end{proof}

\vspace{.2cm}%HACK to fix the page break in the lemma below
\subsection*{Wump Graphs}
\label{sec: reliability}

The line $x=1$ in the Tutte plane, the \emph{reliability line},
is not covered by the above since here $q=0$ holds.
On this line, the Tutte polynomial specializes (up to a closed-form multiplicative factor)
to the {\em reliability polynomial} $R(G;p)$ (with $p=1/y$),
an object studied in algebraic graph theory \cite[Section~15.8]{GoRo01}. %
Given a connected graph~$G$ and a probability~$p$, $R(G;p)$ is the
probability that~$G$ stays connected if every edge independently fails
with probability~$p$. %
For example
$R(
\vcenter{\hbox{$
{\begin{tikzpicture} \path (90+1*72:0.15) node[vertex]
      (v0) {}; \path (90+2*72:0.15) node[vertex] (v1) {}; \path
      (90+3*72:0.15) node[vertex] (v2) {}; \path (90+4*72:0.15)
      node[vertex] (v3) {}; \path (90+5*72:0.15) node[vertex] (v4) {};
      \draw (v0) -- (v1) -- (v2) -- (v3) -- (v4) -- (v0);
\end{tikzpicture}}$}}\,
;\frac{1}{3})
=
Pr(
\vcenter{\hbox{$
{\begin{tikzpicture} \path (90+1*72:0.15) node[vertex]
      (v0) {}; \path (90+2*72:0.15) node[vertex] (v1) {}; \path
      (90+3*72:0.15) node[vertex] (v2) {}; \path (90+4*72:0.15)
      node[vertex] (v3) {}; \path (90+5*72:0.15) node[vertex] (v4) {};
      \draw (v0) -- (v1) -- (v2) -- (v3) -- (v4) -- (v0);
\end{tikzpicture}}$}}
)
+
5Pr(
\vcenter{\hbox{$
{\begin{tikzpicture} \path (90+1*72:0.15) node[vertex]
      (v0) {}; \path (90+2*72:0.15) node[vertex] (v1) {}; \path
      (90+3*72:0.15) node[vertex] (v2) {}; \path (90+4*72:0.15)
      node[vertex] (v3) {}; \path (90+5*72:0.15) node[vertex] (v4) {};
      \draw (v0) -- (v1) -- (v2) -- (v3) -- (v4);
\end{tikzpicture}}$}}
)
=
(\frac{2}{3})^5 + 5\cdot \frac{1}{3}\cdot
(\frac{2}{3})^4 = \frac{112}{243}$.
Note that $R(G;1)=0$ for all connected graphs, so $p=1$ is easy to
evaluate, which we know is also the case (though for less trivial reasons)
for the corresponding limit point $(1,1)$ in the Tutte plane. %

Along the reliability line, weight shift identities take a different
form.
We use deletion--contraction identities to derive the following rules.
They are simple multi-weighted generalizations of
\cite[Section~4.3]{GJ08Tutte}.
\begin{lemma}	\label{lem: stretch}
  Let $G$ be a graph with edge weights given by $\mathbf w : E(G)\to\Q$.

	If $\varphi(G)$ is obtained from $G$ by replacing a single edge
	$e\in E$ with a simple path of $k$ edges $P=\{e_1,...,e_k\}$
	with $\mathbf w(e_i)= w_i$, then
	\[
		\ZREL{\varphi(G)}{\mathbf w} =
		C_P \cdot \ZREL{G}{\mathbf w[e\mapsto w']}
		\,,
	\]
	where
	\[
		\frac{1}{w'} =
		\frac{1}{w_1}+\cdots+\frac{1}{w_k}
		\;\;\;\;\; \text{ and } \;\;\;\;\;
		C_P = \frac{1}{w'}\prod_{i=1}^k w_i
		\,.
	\]
\end{lemma}
Here $\mathbf w[e\mapsto w']$ denotes the function
$\mathbf w':E(G)\to\Q$ that is identical to $\mathbf w$ except at the
point~$e$ where it is~$\mathbf w'(e)=w'$.
\begin{lemma} 	\label{lem: thick}
	If $\varphi(G)$ is obtained from $G$ by replacing a single edge
	$e\in E$ with a bundle of parallel edges $B =\{e_1,\ldots,e_k\}$
	with $\mathbf w(e_i)= w_i$, then
	\[
		\ZREL{\varphi(G)}{\mathbf w} =
		\ZREL{G}{\mathbf w[e\mapsto w']}
		\,,
	\]
	where
	\[
		w' = -1 + \prod_{i=1}^k (1+w_i)
		\,.
	\]
\end{lemma}
\begin{corollary} \label{cor: stretchthick}
	If $\varphi(G)$ is obtained from $G$ by replacing a single edge
	$e\in E$ with a simple path of $k$ edges of constant weight $w$,
	then
	\begin{equation} \label{eq: stretch}
		\ZREL{\varphi(G)}{\mathbf w} =
		kw^{k-1} \cdot \ZREL{G}{\mathbf w[e\mapsto w/k]}
		\,,
	\end{equation}
	and if it is obtained from $G$ by replacing $e\in E$ with a
	bundle of $k$ parallel edges of constant weight $w$, then
	\begin{equation} \label{eq: thick}
		\ZREL{\varphi(G)}{\mathbf w} =
		\ZREL{G}{\mathbf w[e\mapsto (1+w)^k-1]}
		\,.
	\end{equation}
\end{corollary}
These rules are transitive \cite[Lemma~1]{GJ08Tutte}, and so can be freely
combined for more intricate weight shifts. %
We define a class of graph inflations, \emph{Wump inflations}, and
use the above to show that they give rise to distinct weight shifts
along the reliability line of the Tutte polynomial. %
Wump inflations are mildly inspired by $l$-byte numbers, in the
sense that each has associated to it a sequence of length $l$, such
that the lexicographic order of these sequences determines the size of
the corresponding (shifted) weights. %
\begin{definition}[Wump graph]
  For positive integers $i$ (height) and $s$ (width),
  an {\em $(i,s)$-hump} is the graph obtained by identifying all the
  left and all the right endpoints of~$i$ simple paths of length~$s$
  each.
  Given a sequence $S = \langle s_1, s_2, \dots, s_l \rangle$ of~$l$
  positive integers, the {\em Wump graph $\wump{S}$} is the
  graph obtained by concatenating $l$ humps at their endpoints,
  where the $i$-th hump is an $(i,s_i)$-hump, i.e., its height
  is~$i$ and its width is~$s_i$.

	{\centering
	\hfill
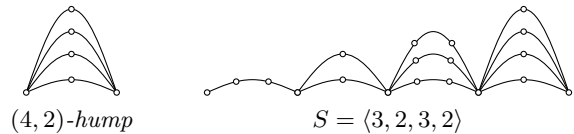
\par
	\begin{tikzpicture}[scale=0.85]
	%single hump
	\draw (0-1.4,0) sin (0.7-1.4,0.2) cos (1.4-1.4,0);
	\draw (0-1.4,0) sin (0.7-1.4,0.6) cos (1.4-1.4,0);
	\draw (0-1.4,0) sin (0.7-1.4,0.95) cos (1.4-1.4,0);
	\draw (0-1.4,0) sin (0.7-1.4,1.3) cos (1.4-1.4,0)
	plot[only marks, mark=*, mark options={fill=white}, mark size=1.2pt]
	coordinates{ (0-1.4,0) (0.7-1.4,0.6) (0.7-1.4,0.2)
	(0.7-1.4,0.95) (0.7-1.4,1.3) (1.4-1.4,0)};
	\coordinate (label) at (0.7-1.4,-0.8);
	\node [above] at (label) {\footnotesize $(4,2)$-hump};

	%Bs: hump 1
	\draw (1*1.4,0) sin (1*1.4+0.7,0.2) cos (1*1.4+1.4,0)
	plot[only marks, mark=*, mark options={fill=white}, mark size=1.2pt]
	coordinates{(1*1.4+0,0) (1*1.4+0.45,0.17) (1*1.4+0.95,0.17) (1*1.4+1.4,0) };
	%Bs: hump 2
	\draw (2*1.4,0) sin (2*1.4+0.7,0.2) cos (2*1.4+1.4,0);
	\draw (2*1.4,0) sin (2*1.4+0.7,0.6) cos (2*1.4+1.4,0)
	plot[only marks, mark=*, mark options={fill=white}, mark size=1.2pt]
	coordinates{(2*1.4+0,0) (2*1.4+0.7,0.2) (2*1.4+0.7,0.6) (2*1.4+1.4,0)};
	%Bs: hump 3
	\draw (3*1.4,0) sin (3*1.4+0.7,0.2) cos (3*1.4+1.4,0);
	\draw (3*1.4,0) sin (3*1.4+0.7,0.6) cos (3*1.4+1.4,0);
	\draw (3*1.4,0) sin (3*1.4+0.7,0.95) cos (3*1.4+1.4,0)
	plot[only marks, mark=*, mark options={fill=white}, mark size=1.2pt]
	coordinates{
	(3*1.4,0) (3*1.4+1.4,0)
	(3*1.4+0.45,0.17) (3*1.4+0.95,0.17)
	(3*1.4+0.43,0.50) (3*1.4+0.97,0.50)
	(3*1.4+0.41,0.77) (3*1.4+0.99,0.77)
	};
	%Bs: hump 4
	\draw (4*1.4,0) sin (4*1.4+0.7,0.2) cos (4*1.4+1.4,0);
	\draw (4*1.4,0) sin (4*1.4+0.7,0.6) cos (4*1.4+1.4,0);
	\draw (4*1.4,0) sin (4*1.4+0.7,0.95) cos (4*1.4+1.4,0);
	\draw (4*1.4,0) sin (4*1.4+0.7,1.3) cos (4*1.4+1.4,0)
	plot[only marks, mark=*, mark options={fill=white}, mark size=1.2pt]
	coordinates{
	(4*1.4+0,0) (4*1.4+0.7,0.2) (4*1.4+0.7,0.6)
	(4*1.4+0.7,0.95) (4*1.4+0.7,1.3) (4*1.4+1.4,0)
	};
	%label
	\coordinate (label) at (3*1.4,-0.8);
	\node [above] at (label)
	{\footnotesize $S=\langle 3,2,3,2 \rangle$};
	\end{tikzpicture}
	\par}
  \noindent The number~$l$ is the {\em length} of the Wump graph~$\wump{S}$.
\end{definition}

Inflating a graph by a Wump graph shifts the weights on the
reliability line as follows.
\begin{lemma}	\label{lem: wumpweights}%
  For any graph~$G$ with $m$ edges,
  any sequence $S = \langle s_1, s_2, \dots, s_l \rangle$ of positive
  integers, and any non-zero rational number~$w$, we have
  \begin{equation*}
		\ZREL{\inflate{G}{\wump{S}}}{w} =
		C_S^m \cdot \ZREL{G}{w_S}
		\,,
  \end{equation*}
	where
  \begin{equation} \label{eq: wump weight move}%
		\frac{1}{w_S} =
		\sum_{i=1}^{l} \frac{1}{(1+w/s_i)^i-1}
		\;\;\; \text{ and }	\;\;\;
		C_S =
		\frac{1}{w_S} \cdot
		\prod_{i=1}^l w^{(s_i-1)i}
		\left( (w+s_i)^i-s_i^i \right)
		\,.%
  \end{equation}
\end{lemma}
\begin{proof}
  We start with \inflate{G}{\wump{S}} and consider the effect that
  replacing one of the~$m$ canonical copies of $\wump{S}$ with a
  single edge~$e$ has.
  We show that, with~$\varphi$ denoting this operation,
	\begin{equation} \label{eq: flatwump}
		\ZREL{\inflate{G}{\wump{S}}}{w} =
		C_S \cdot
		\ZREL{\varphi (\inflate{G}{\wump{S})}}
		{\mathbf w[e\mapsto w_S]}
		\,,
	\end{equation}
	where $w_S$ has the above form, and $\mathbf w$ has the old value $w$
  on all unaffected edges. The lemma then follows by successively
  applying~$\varphi$ to each canonical copy of~$\wump{S}$ in
  \inflate{G}{\wump{S}}.

	The first step towards transforming a Wump graph (say,
	\begin{tikzpicture}[scale=0.38]
	%hump 1
	\draw (0*1.4,0) sin (0*1.4+0.7,0.2) cos (0*1.4+1.4,0)
	plot[only marks, mark=*, mark options={fill=white}, mark size=2.5pt]
	coordinates{(0*1.4+0,0) (0*1.4+0.45,0.17)
	(0*1.4+0.95,0.17) (0*1.4+1.4,0) };
	%hump 2
	\draw (1*1.4,0) sin (1*1.4+0.7,0.2) cos (1*1.4+1.4,0);
	\draw (1*1.4,0) sin (1*1.4+0.7,0.6) cos (1*1.4+1.4,0)
	plot[only marks, mark=*, mark options={fill=white}, mark size=2.5pt]
	coordinates{(1*1.4+0,0) (1*1.4+0.7,0.2) (1*1.4+0.7,0.6) (1*1.4+1.4,0)};
	%hump 3
	\draw (2*1.4,0) sin (2*1.4+0.7,0.2) cos (2*1.4+1.4,0);
	\draw (2*1.4,0) sin (2*1.4+0.7,0.6) cos (2*1.4+1.4,0);
	\draw (2*1.4,0) sin (2*1.4+0.7,0.95) cos (2*1.4+1.4,0)
	plot[only marks, mark=*, mark options={fill=white}, mark size=2.5pt]
	coordinates{ (2*1.4,0) (2*1.4+0.7,0.2) (2*1.4+0.7,0.6)
	(2*1.4+0.7,0.95)  (2*1.4+1.4,0)};
	\end{tikzpicture})
  into a single edge, consists of contracting the paths of the humps
  to a single edge each.
  For the $i$-th hump, this is just the inverse of an
  $s_i$-stretching applied to each of the~$i$ paths.
  By~\eqref{eq: stretch} of Corollary~\ref{cor: stretchthick},
  this ``unstretching'' gives a factor $(s_iw^{s_i-1})^i$ to the
  polynomial, and each edge in the resulting $(i,1)$-hump receives a
  weight of $w/s_i$ in the modified graph.
  Repeating this process for every hump simplifies the Wump graph
  into a Wump graph of length~$l$ that is generated by a sequence
  of~$1$s
	(\begin{tikzpicture}[scale=0.35]
	%hump 1
	\draw (0*1.4,0) sin (0*1.4+0.7,0.2) cos (0*1.4+1.4,0)
	plot[only marks, mark=*, mark options={fill=white}, mark size=2.5pt]
	coordinates{(0*1.4+0,0) (0*1.4+1.4,0) };
	%hump 2
	\draw (1*1.4,0) sin (1*1.4+0.7,0.2) cos (1*1.4+1.4,0);
	\draw (1*1.4,0) sin (1*1.4+0.7,0.6) cos (1*1.4+1.4,0)
	plot[only marks, mark=*, mark options={fill=white}, mark size=2.5pt]
	coordinates{(1*1.4+0,0) (1*1.4+1.4,0)};
	%hump 3
	\draw (2*1.4,0) sin (2*1.4+0.7,0.2) cos (2*1.4+1.4,0);
	\draw (2*1.4,0) sin (2*1.4+0.7,0.6) cos (2*1.4+1.4,0);
	\draw (2*1.4,0) sin (2*1.4+0.7,0.95) cos (2*1.4+1.4,0)
	plot[only marks, mark=*, mark options={fill=white}, mark size=2.5pt]
	coordinates{(2*1.4,0) (2*1.4+1.4,0)};
	\end{tikzpicture}).
  Let $\phi(\inflate{G}{\wump{S}})$ denote the graph in which one
  Wump graph has been simplified.
  By transitivity, we have the weight shift
	\[
		\ZREL{\inflate{G}{\wump{S}}}{w} =
		\left(\prod_{i=1}^l (s_iw^{s_i-1})^i \right)
		\cdot
		\ZREL{\phi(\inflate{G}{\wump{S}})}{\mathbf w'}
		\,,
	\]
  where $\mathbf w'$ takes the value $w/s_i$ on every edge of the
  $i$th hump of the simplified Wump graph, and the old value $w$
  outside the simplified Wump graph.
  Next, we successively replace each of its $(i,1)$-humps by a
  single edge to get a simple path
	(\begin{tikzpicture}[scale=0.35]
	\draw (0,0) -- (1.4,0) -- (2*1.4,0) -- (3*1.4,0)
	plot[only marks, mark=*, mark options={fill=white}, mark size=3pt]
	coordinates{(0*1.4,0) (1*1.4,0) (2*1.4,0) (3*1.4,0)};
	\end{tikzpicture})
	of length~$l$.
  This transformation is just an ``unthickening'' of each
  $(i,1)$-hump, and from~\eqref{eq: thick} of
  Corollary~\ref{cor: stretchthick} we know that it does not produce
  any new factors for the polynomial, but the weight of the $i$th edge
  in this path becomes
	\[
		w_i=
		\left(
		1+w/s_i
		\right)^i - 1
		\,.
	\]
	Finally, we compress the path into a single edge $e$.
  Then the claim in~\eqref{eq: flatwump} follows by a single
  application of Lemma~\ref{lem: stretch}.
\end{proof}
We now show that Wump inflations provide a rich enough class of
weight shifts.
The ranges of~$w$ for which we prove this is general enough to allow
for interpolation on the whole reliability line, and we make no
attempt at extending the ranges. %
In the following lemma, we use the definition of $w_S$ from
\eqref{eq: wump weight move}.%

\begin{lemma} \label{lem: lever}
  Let $w$ be a rational number with $w\in(-1,0)$ or $w\in(9,\infty)$.
  For all integers $m\geq 1$, there exist sequences $S_0,\dots,S_m$
  of positive integers such that
  \begin{enumerate}[(i)]
    \item\label{wump conclusion w_s total size}
      $|E(\wump{S_i})| \leq O(\log^2 m)$ for all $i$, and
    \item\label{wump conclusion w_s distinct}
      $w_{S_i} \neq w_{S_j}$ for all $i\neq j$.
  \end{enumerate}
  Furthermore, the sequences $S_i$ can be computed in time polynomial
  in~$m$.
\end{lemma}
\begin{proof}
  We consider the set of sequences
  $S=\langle s_1,\dots,s_l\rangle$ of length $l=r\log(m+1)$,
  with $s_i \in \{2,3\}$ for all $i$ which are positive integer
  multiples of $r$, and $s_i=2$ for all other~$i$.
  Here $r$ is a positive integer and will be chosen later, only
  depending on~$w$.
  Since $r$ is a constant, this construction
  satisfies~\iref{wump conclusion w_s total size}.

  Now consider any two distinct sequences $S=\langle s_i\rangle$ and
  $T=\langle t_i\rangle$.
  To show~\iref{wump conclusion w_s distinct},
  we consider the difference
	\[
		\Delta = \frac{1}{w_S}-\frac{1}{w_T} \,,
	\]
  and show that $\Delta\neq0$.

  Using Lemma~\ref{lem: wumpweights} we get a sum expression for
  $\Delta$.
	\begin{equation}\label{eqn: Deltasum}%
    \begin{split}
		\Delta
		& =
		\sum_{i=1}^{l} \frac{1}{(1+w/s_{i})^{i}-1}
		-
		\sum_{i=1}^{l} \frac{1}{(1+w/t_{i})^{i}-1}
    \\
    & =
		\sum_{i=1}^{l} g\left((1+w/s_{i})^{i}\right)
		-
		\sum_{i=1}^{l} g\left((1+w/t_{i})^{i}\right)\,,%
    \end{split}
	\end{equation}
  where $g$ is the function $g(x)=\frac{1}{x-1}$.
  This function is negative and strictly decreasing on $(0,1)$ and
  positive and strictly decreasing on $(1,\infty)$.
  It is convenient to choose $a,b\in\{(1+w/3),(1+w/2)\}$ so that $a<b$.
  By the monotonicity of $g$, we have $g(a^i)>g(b^i)$ for all
  positive~$i$.

  \textbf{Case 1:} $w>9$.
  Here we have $a=(1+w/3)$ and $b=(1+w/2)$.
  We set $r=1$ and let~$k$ be the smallest index for which the
  sequences differ, i.e., $s_k\neq t_k$.
  We assume w.l.o.g.\ that $s_{k}=3$ and $t_{k}=2$, otherwise we
  exchange the roles of~$S$ and~$T$.
  In~\eqref{eqn: Deltasum}, terms of the sum for $i<k$ cancel.
  The terms corresponding to $i=k$ are $g(a^k)-g(b^k)>0$.
  We apply the monotonicity of $g$ to the terms for $i>k$, which allows
  us to lower bound $\Delta$ as follows.
  \begin{equation*}
		\Delta \geq
		g(a^k) + \sum_{i=k+1}^{l} g(b^{i})
		-
		g(b^k) - \sum_{i=k+1}^{l} g(a^{i})
    =
		f(a) - f(b)
		\,,
  \end{equation*}
	where
	\begin{equation} \label{eq: f(x)}
		f(x) =
    g(x^k) - \sum_{i=k+1}^{l} g(x^i)
    =
		\frac{1}{x^k-1} -
		\sum_{i=k+1}^{l} \frac{1}{x^i-1}
		\,.
	\end{equation}
  We now claim that $f$ is strictly decreasing in $(4,\infty)$.
  This implies $\Delta>0$ since $w>9$ guarantees $a,b>4$, and we
  get $\Delta \geq f(a)-f(b)>0$. %
  To prove the claim, we show that the derivative of $f$ is negative
  on $(4,\infty)$. This is a routine calculation, but we include it here for
  completeness. We have
  \begin{equation} \label{eq: f'(x)}
    f'(x) =
    -\frac{kx^{k-1}}{(x^k-1)^2} +
    \sum_{i=k+1}^{l} \frac{ix^{i-1}}{(x^i-1)^2}
    \,.
  \end{equation}
  The terms of the sum here, let us call them $T_i(x)$, satisfy
  \begin{equation*}
    T_i(x) > 2 \cdot T_{i+1}(x)
  \end{equation*}
  for all $i$ and all $x>4$.
  To see this, note that the inequality is equivalent to
  \begin{equation*}
    2\left(1+\frac{1}{i}\right)x
    < \left(x+\frac{x-1}{x^i-1}\right)^2
    \,.
  \end{equation*}
  This statement is true for all reals $x>4$ and all positive integers
  $i$ since
  then we have that $\text{LHS} \leq 4x < x^2 \leq \text{RHS}$.
  Thus, for $x>4$, we have
  \begin{equation*}
    f'(x) <
    \frac{kx^{k-1}}{(x^k-1)^2}
    \left( -1 + \sum_{i=k+1}^{l} \frac{1}{2^{i-k}} \right)
    < 0
    \,.
  \end{equation*}

  \textbf{Case 2:} $w\in (-1,0)$.
  Here we have $a=(1+w/2)$ and $b=(1+w/3)$.
  We choose $r$ to be a positive integer that satisfies
  $b^r<\frac{1}{4}$.
  Let $rk$ be the smallest index for which the sequences differ, i.e.,
  $s_{rk}\neq t_{rk}$.
  We assume w.l.o.g.\ that $s_{rk}=3$ and $t_{rk}=2$, otherwise we
  exchange the roles of $S$ and $T$.
  In~\eqref{eqn: Deltasum}, terms of the sum for $i<rk$ cancel, and so
  do terms for those $i$'s which are not integer multiples of $r$.
  The terms corresponding to $i=rk$ are $g(b^{rk})-g(a^{rk})<0$.
  We apply the monotonicity of $g$ to the remaining terms for $i>rk$,
  which allows us to upper bound $\Delta$ as follows.
  \begin{equation*}
		\Delta
    \leq
    g(b^{rk})
    + \sum_{i=k+1}^{l/r} g(a^{ri})
		- g(a^{rk})
    - \sum_{i=k+1}^{l/r} g(b^{ri})
	\end{equation*}
  For $x\in(0,1)$, we can expand $g(x)$ into the geometric series
  \begin{equation*}
  g(x)
  =\frac{1}{x-1}
  =-\sum_{j=0}^\infty x^j\,.
  \end{equation*}
  Applying this representation to our estimate for $\Delta$ and
  rearranging terms, we arrive at
	\begin{align*}
		\Delta
    &\leq
    \sum_{j=0}^\infty
    \left(
    (a^{rj})^k
    - (b^{rj})^k
    + \sum_{i=k+1}^{l/r} \left((b^{rj})^i - (a^{rj})^i\right)
    \right)
    =
    \sum_{j=0}^\infty \left( F(a^{rj})-F(b^{rj}) \right)
    \,,
	\end{align*}
  where $F$ is the function
  \begin{equation*}
    F(y)=y^k - \sum_{i=k+1}^{l/r} y^i\,.
  \end{equation*}
  We claim that $F$ is strictly increasing on $(0,\frac{1}{4})$.
  This, together with the fact that~$r$ is chosen such that
  $a^{rj},b^{rj}\in(0,\frac{1}{4})$ for all
  positive integers $j$, implies $\Delta<0$, because then
  $F(a^{rj})-F(b^{rj})<0$ for $j\geq 1$, and for $j=0$ the term is $0$.
  To prove the claim we show that the derivative of $F$ is positive on $(0,\frac{1}{4})$.
  Again, we give the details here for completeness.
  We have
  \begin{equation*}
    F'(y)=k y^{k-1} - \sum_{i=k+1}^{l/r} i y^{i-1}\,,
  \end{equation*}
  and obtain $F'(y)>0$ from the following calculation, using the
  fact that $y\in(0,\frac{1}{4})$.
  \begin{align*}
    (k y^{k-1})^{-1} \cdot \sum_{i=k+1}^{l/r} i y^{i-1}
    &=
    \sum_{i=k+1}^{l/r} \frac{i}{k} y^{i-k}
    =
    \sum_{i=1}^{l/r-k} \left(1+\frac{i}{k}\right) y^{i}\\
    &\leq
    \sum_{i=1}^{l/r-k} \left(1+i\right) y^{i}
    \leq
    \sum_{i=1}^{\infty} y^{i} + \sum_{i=1}^{\infty} i y^{i}\\
    &=
    \frac{1}{1-y} - 1 + \frac{y}{(1-y)^2}
    \leq
    \frac{4}{3} - 1 + \frac{4}{9}<1\,.%
  \end{align*}
\end{proof}

\subsection*{Points on the Reliability Line}

We prove
Theorem~\ref{thm: Tutte main result}\iref{thmi: Tutte reliability}.%

\begin{proposition}\label{prop: reliability}
  Let $w\neq 0$ be a rational number.
  If \cETH{} holds, then \ZREL{G}{w} for a given simple graph $G$ cannot
  be computed in time $\exp(o(m/\log^2m))$.
\end{proposition}

\begin{proof}
  If $w<0$, we can pick a positive integer~$k$ big enough such that
  \begin{equation*}
    w' := w/k > -1\,.
  \end{equation*}
  This weight shift corresponds to the $k$-stretch of $G$
  (Corollary~\ref{cor: stretchthick}).
  On the other hand, if $w > 0$, we can pick a positive integer~$k$
  such that
  \begin{equation*}
    w' := (w/2+1)^k-1 > 9\,.
  \end{equation*}
  This is the weight shift that corresponds to the $2$-stretch of the
  $k$-thickening of $G$ (Corollary~\ref{cor: stretchthick}).
  In any case we can compute $Z(G;w',q)$ from $Z(G';w,q)$.
  The graph remains simple after any of these transformations, and the
  number of edges is only increased by a constant factor of at
  most~$2k$.

  By the above, we can assume w.l.o.g.\ that $w\in(-1,0)$ or $w>9$.
  We use Lemma~\ref{lem: lever} to construct $m+1$
  Wump graphs~\wump{S} whose corresponding weight shifts $w_S$ are
  all distinct by property~\iref{wump conclusion w_s distinct} of
  Lemma~\ref{lem: lever}.
  By Lemma~\ref{lem: wumpweights},
  we can compute the values \ZREL{G}{w_S} from
  \ZREL{\inflate{G}{\wump{S}}}{w}, i.e.,
  we get evaluations of $v\mapsto \ZREL{G}{v}$ at $m+1$ distinct
  points.
  Since the degree of this polynomial is $m$, we obtain its
  coefficients by interpolation.
  By Proposition~\ref{prop: hyp}, these coefficients cannot be
  computed in time $\exp(o(m))$ under \cETH.
  By Lemma~\ref{lem: lever}\iref{wump conclusion w_s total size},
  each \inflate{G}{\wump{S}} has at most $\Oh(m \log^2 m)$ edges,
  which implies that \ZREL{G}{w} for given $G$ cannot be computed in time
  $\exp(o(m/\log^2 m))$ as claimed.
\end{proof}

\section{Conclusion and Further Work}
Our results for the Tutte polynomial leave open the line $y=1$ except
for the point $(1,1)$, even in the case of multigraphs. %
That line corresponds to counting the number of forest weighted by the
number of edges, i.e.,
$T(G;1+1/w,1) \sim F(G;w)=\sum_{\text{forests $F$}} w^{|F|}$. %
Thickening and Theta inflation, with the analysis in the proof of
Lemma~\ref{lem: wumpweights}, suffice to show that every point is as
hard as computing the coefficients of $F(G;w)$, without increasing the
number of vertices for multigraphs and with an increase in the number
of edges by a factor of $O(\log^2 m)$ in the case of simple graphs. %
However, we do not know whether computing those coefficients requires
exponential time under \cETH. %
And of course, it would be nice to improve our conditional lower
bounds $\exp(\Omega(n/\poly\log n))$ to match the corresponding upper
bounds $\exp(O(n))$.

\subsection*{Acknowledgements}

The authors are grateful to
Andreas Bj\"orklund,
Leslie Ann Goldberg, and
Dieter van Melkebeek
for valuable comments.

Wump graphs are named for a fictional creature notable for its number
of humps, which appears in the American children's book ``One Fish Two
Fish Red Fish Blue Fish'' by Dr.~Seuss;
the name was suggested by Prasad Tetali.

\printbibliography

\newpage\appendix

\section{The Sparsification Lemma}

\label{sec: counting sparsification}%

Sparsification is the process of reducing the density of graphs,
formulas, or other combinatorial objects, while some properties of the
objects like the answer to a computational problem are preserved.

The objective of sparsification is twofold. %
From an algorithmic perspective, efficient sparsification procedures
can be used as a preprocessing step to make input instances sparse and
thus possibly simpler and smaller, such that only the core information
about the input remains. %
In the literature, such applications of sparsification procedures are
called kernelizations.
From a complexity-theoretic point of view, sparsification is a tool to
identify those instances of a problem that are computationally the
hardest. %
If an \cc{NP}-hard problem admits efficient sparsification, the
hardest instances are sparse.

In the context of the exponential time hypothesis, the sparsification
lemma provides a way to show that the hardest instances of
\pp{$d$-Sat} are sparse and thus the parameter $n$ can be replaced
with $m$ in the statement of the exponential time hypothesis.  %
The following is the sparsification lemma as formulated
in~\cite[Lemma~16.17]{FG06}.
\begin{lemma}[Sparsification Lemma]\label{counting sparsification lemma}%
  Let $d \geq 2$. %
  There exists a computable function $f:\N^2\to\N$ such that for every
  $k \in \N$ and every $d$-CNF formula $\gamma$ with $n$ variables, we
  can find a formula
  \begin{equation*}
    \beta = \bigvee_{i \in [t]} \gamma_i
  \end{equation*}
  such that:
  \begin{enumerate}[(1)]
    \item $\beta$ is equivalent to $\gamma$ (ie., they have the same
    satisfying assignments),\label{spars-lemma-prop1}
    \item $t \leq 2^{n/k}$, and
    \item the $\gamma_i$ are $d$-CNF formulas in which each variable
    occurs at most $f(d,k)$ times.
  \end{enumerate}
  Furthermore, $\beta$ can be computed from $\gamma$ and $k$
  in time $t \cdot \poly(n)$.
\end{lemma}
We sketch below a small modification in the proof of the
sparsification lemma that allows us to replace
\iref{spars-lemma-prop1} with the condition
\begin{enumerate}[\it (1$'$)]
  \item
    $\op{sat}(\gamma)= \dot\bigcup_i \op{sat}(\gamma_i)$\,,
\end{enumerate}
where $\op{sat}(\varphi)$ is the set of assignments that satisfy the
formula~$\varphi$. %
That is, not only is $\beta$ equivalent to $\gamma$, it even holds
that every satisfying assignment of $\beta$ satisfies exactly one
$\gamma_i$.
In particular, \textit{(1$'$)} implies
$\pp{\#Sat}(\gamma)=\sum_i \pp{\#Sat}(\gamma_i)$, which means that the
sparsification lemma can be used for the counting version of
\pp{$3$-Sat}.
\begin{proof}[sketch]
  We adapt the terminology of~\cite[Proof of Lemma~16.17]{FG06} and we
  follow their construction precisely, except for a small change in
  the sparsification algorithm. %
  When the algorithm decides to branch for a CNF-formula~$\gamma$ and
  a flower $\alpha=\left\{ \delta_1,\dots,\delta_p \right\}$, the
  original algorithm would branch on the two formulas
  \begin{align*}
  \gamma_{\text{heart}}^\alpha
  &= \gamma \setminus \left\{ \delta_1,\dots,\delta_p \right\}
  \cup \left\{ \delta\right\},\\
  \gamma_{\text{petals}}^\alpha
  &= \gamma \setminus \left\{ \delta_1,\dots,\delta_p \right\}
  \cup \left\{ \delta_1 \setminus \delta,\dots,\delta_p \setminus\delta\right\}.
  \end{align*}
  We modify the branching on the petals to read
  \[\gamma_{\text{petals}}^\alpha
  = \gamma \setminus \left\{ \delta_1,\dots,\delta_p \right\} \cup
  \left\{ \delta_1 \setminus \delta,\dots,\delta_p
    \setminus\delta\right\} \cup \big\{ \,\left\{\neg l\right\} \colon
  l \in \delta\, \big\}\,.\] This way, the satisfying assignments
  become disjoint: %
  In each branching step, we guess whether the heart contains a
  literal set to true, or whether all literals in the heart are set to
  false and each petal contains a literals set to true.

  Now we have that, for all CNF-formulas~$\gamma$, all
  assignments~$\sigma$ to the variables of~$\gamma$, and all
  flowers~$\alpha$ of~$\gamma$,
  \begin{enumerate}[(i)]
    \item $\sigma$ satisfies $\gamma$ if and only if
      $\sigma$ satisfies $\gamma_{\text{heart}}^\alpha
      \vee \gamma_{\text{petals}}^\alpha$,
      and
    \item
      $\sigma$ does not satisfy $\gamma_{\text{heart}}^\alpha$
      or $\sigma$ does not satisfy $\gamma_{\text{petals}}^\alpha$.
  \end{enumerate}
  By induction, we see that at the end of the algorithm,
  \begin{enumerate}[(i)]
    \item $\sigma$ satisfies $\gamma$ if and only if $\sigma$
      satisfies some $\gamma_i$, and
    \item
      $\sigma$ satisfies at most one $\gamma_i$.
  \end{enumerate}
  This implies that
  $\op{sat}(\gamma)=\dot\bigcup_{i \in [t]}\op{sat}(\gamma_i)$.

  Notice that our new construction adds at most~$n$ clauses of
  size~$1$ to the formulas~$\gamma_i$ compared to the old one. %
  Furthermore, our construction does not make~$t$ any larger because
  the \textsc{REDUCE}-step removes all clauses that properly
  contain~$\left\{ \neg l \right\}$ and thus these unit clauses never
  appear in a flower.
\end{proof}

\begin{proof}[of Theorem~\ref{thm: counting_sparsification_essence}]
  For all integers $d\geq 3$ and $k\geq 1$, the sparsification lemma
  gives an oracle reduction from \pp{\#$d$-Sat} to \pp{\#$d$-Sat}
  that, on input a formula~$\gamma$ with $n$ variables, only queries
  formulas with $m'=O(n)$ clauses, such that the reduction runs in
  time~$\exp(O(n/k))$. %
  Now, if for every $c>0$ there is an algorithm for \pp{\#$d$-Sat}
  that runs in time~$\exp(c m)$, we can combine this algorithm and the
  above oracle reduction to obtain an algorithm for \pp{\#$d$-Sat}
  that runs in time $\exp(O(n/k) + c\cdot m')= \exp(O(n/k) + c \cdot
  O(n))$.  %
  Since this holds for all small $c>0$ and large~$k$, we have for
  every $c'>0$ an algorithm for \pp{\#$d$-Sat} running in
  time~$\exp(c'\cdot n)$. %
  This proves that for all $d\geq 3$, \pp{\#$d$-Sat} can be solved in
  variable-subexponential time if and only if it can be solved in
  clause-subexponential time.

  It remains to show that \pp{\#$d$-Sat} reduces to \pp{\#$3$-Sat}.
  We transform an instance~$\varphi$ of \pp{\#$d$-Sat} into an
  instance $\varphi'$ of \pp{\#$3$-Sat} that has the same number of
  satisfying assignments.
  The formula $\varphi'$ is constructed as in the standard
  width-reduction for $d$-CNF formulas, i.e., by introducing a
  constant number of new variables for every clause of $\varphi$. %
  Thus, since the number of clauses of $\varphi'$ is $O(m)$, any
  clause-subexponential algorithm for \pp{\#$3$-Sat} implies a
  clause-subexponential algorithm for \pp{\#$d$-Sat}.
\end{proof}

\section{Parameterized Complexity}
\label{sec: parameterized}

Our hypothesis \cETH{} relates to parameterized complexity, which is a
branch of computational complexity that considers problems in terms of
two parameters~$n$ and~$k$.
Of special interest in that field are problems that have algorithm
whose running times are of the form $f(k)\poly(n)$ for some computable
function~$f$.
Such problems are called fixed parameter tractable, or \cc{FPT}.

\citeN{FG04} introduce the class $\cc{\#W[1]}$ of parameterized
counting problems.
This class is characterized by complete problems such as computing the
number of cliques of size~$k$ or computing the number of simple paths
of length~$k$ in an $n$-vertex graph.
Implicitly, \citeN{FG04} show that these problems are not
fixed-parameter tractable under \cETH.
\begin{theorem}[Flum and Grohe]
  If \cETH{} holds, then $\cc{\#W[1]} \neq \cc{FPT}$.
\end{theorem}
The latter is only an implication and, as in the case of decision
problems, we do not know whether the two claims are equivalent.
For a claim that is equivalent to a uniform variant of \cETH{}, we can
follow a construction due to \citeN{g2003cutting}.
They consider the following problem:
\begin{quote}\small
  \begin{description}
  \item[Name] \pp{\#{}Mini-$3$-Sat}
  \item[Input] Integers $k$ and $n$; a $3$-CNF formula~$\varphi$ with
    at most $k\log n$ clauses.
  \item[Output] The number of satisfying assignments of $\varphi$.
  \end{description}
\end{quote}
Without explicit reference to \ETH{}, Downey et
al.~\cite{g2003cutting} (based on ideas of \citeN{CaiJuedes}) prove
that the decision version of this
problem is equivalent to a uniform variant of \ETH{}.
By a straightforward modification of their reduction, one can
establish the following equivalence (see also
\cite[chapter~16]{FG06}).
\begin{theorem}[Downey et al.]
  The following two statements are equivalent.
  \begin{enumerate}[(i)]
    \item There is no computable function $T(n) \leq 2^{o(n)}$ such
      that \pp{\#{}$3$-Sat} has a deterministic algorithm
      that runs in time $T(n)$ for $n$-variable formulas.
    \item \mbox{\pp{\#{}Mini-$3$-Sat} $\notin\mathrm{FPT}$}.
  \end{enumerate}
\end{theorem}

\section{Hardness of 3-Colouring and 3-Terminal MinCut}
\label{app: standard hardness results}%

The purpose of this section is to show that the standard reductions
from \pp{$3$-Sat} to \pp{$3$-Colouring},
\pp{NAE-$3$-Sat}, \pp{MaxCut}, and \pp{$3$-Ter\-mi\-nal Min\-Cut}
computationally preserve the number of solutions and increase the
number of clauses or edges of the instances by at most a constant
factor. %
This implies that the corresponding counting problems cannot be
computed in clause-subexponential or edge-subexponential time
unless~\cETH{} fails.
\begin{theorem}\label{thm: standard reductions}
  \hspace{-1.3pt}
  The problems \pp{\#NAE-$3$-Sat}, \pp{\#MaxCut}, \pp{\#$3$-Terminal
  MinCut}, and \pp{\#$3$-Colouring}
  cannot be deterministically computed in time
  $\exp(o(m))$ unless \cETH{} fails.
\end{theorem}
In the following, we formally define the problems, sketch the standard
\cc{NP}-hardness reductions, and provide their analyses as needed to
prove Theorem~\ref{thm: standard reductions}.
For the purposes of this section, \emph{polynomial-time reductions}
between counting problems are oracle reductions that make at most one
query.
The reductions we sketch need not be parsimonious, that is, they map
instances of one problems to instances of another problem (which they
query), but the number of solutions need not be exactly equal.
In fact, there is no parsimonious reduction from \pp{\#$3$-Sat} or
\pp{\#NAE-$3$-Sat} to \pp{\#MaxCut} since every graph has at least one
maximum cut while not every formula is satisfiable.
Similarly, reductions from \pp{\#$3$-Sat} to \pp{\#$3$-Terminal
MinCut} cannot be parsimonious.

\subsection*{Not-all-equal-Sat}
We show that counting the number of all not-all-equal assignments is
hard even for the promise problem in which we only have inputs with at
least one such assignment.
A truth assignment is a \emph{not-all-equal assignment} if all
constraints $\{a,b,c\} \in \varphi$ contain a true \emph{and} a false
truth value.
Formally, we use the following promise version of \pp{\#NAE-$3$-Sat}.
\begin{quote}\small
\begin{description}
  \item[Name] \pp{\#NAE-$3$-Sat}$^+$
  \item[Input] $3$-CNF formula~$\varphi$ with at least one
    not-all-equal assignment.
  \item[Output] The number of not-all-equal assignments.
\end{description}
\end{quote}
\begin{lemma}
  There is a polynomial-time reduction from \pp{\#$3$-Sat} to
  \pp{\#NAE-$3$-Sat}$^+$ that maps formulas with $m$ clauses to formulas
  with $O(m)$ clauses.
%  Furthermore, the reduction generates only formulas that have at
%  least one not-all-equal assignment.
\end{lemma}
\begin{proof}
  Let~$\psi$ be a $3$-CNF formula with $n$ variables and $m$
  clauses. %
  To fulfil the promise, we first plant a satisfying
  assignment using a popular homework assignment.
  We obtain a $3$-CNF formula $\varphi$ with $O(m)$ variables and
  clauses such that $\pp{\#Sat}(\varphi) = \pp{\#Sat}(\psi)+1$.

  To construct the instance $\varphi'$ to \pp{NAE-$3$-Sat}, we introduce a
  new variable $x$ for every trivariate clause $(a\vee b\vee c)$ of
  $\varphi$, and we replace that clause with
  \[
  (x\vee \overline{a}) \wedge (x\vee \overline{b}) \wedge (\overline{x}\vee
  a\vee b) \wedge (x\vee c)\,.
  \]
  These clauses force $x$ to have the same value as $a\vee b$ in any
  satisfying assignment. %
  It can be checked that these clauses are satisfied exactly if the
  original clause was satisfied and moreover that the trivariate
  clause is never all-false or all-true. %
  In total, we increased the number of clauses four-fold without
  changing the number of satisfying assignments.

  Finally, introduce a single fresh variable $z$ and add this variable
  (positively) to every mono- and bivariate clause. %
  It is well-known that this modification turns $\varphi'$ into an
  instance~$\varphi''$ of \pp{NAE-$3$-Sat}~\cite[Theorem~9.3]{Papa}:
  The not-all-equal assignments of~$\varphi''$ are exactly the
  satisfying assignments of $\varphi'$ (if $z$ is set to false) or
  their complements (if~$z$ is set to true).

  The reduction computes~$\varphi''$ from $\psi$ in polynomial time,
  $\varphi''$ has at most $O(m)$ clauses, and we have
  $\pp{\#NAE-$3$-Sat}(\varphi'') = 2\cdot (\pp{\#Sat}(\psi)+1)$.
\end{proof}

\subsection*{Maximum Cut}
A \emph{cut} is a set $C\subseteq V(G)$ and its \emph{size} is the
number $|E(C,\overline C)|$ of edges of~$G$ that cross the cut.
A \emph{maximum cut} is a cut $C\subseteq V(G)$ of maximum size.
\begin{quote}\small
\begin{description}
  \item[Name] $\pp{\#MaxCut}$
  \item[Input] Simple undirected graph $G$.
  \item[Output] The number of maximum cuts.
\end{description}
\end{quote}
\citeN[Lemma~13]{JerrumSinclair} modify a reduction of
\citeN[Theorem~1.1 and Theorem~1.2]{GareyJohnsonStockmeyer} to show
\cc{\#P}-hardness of this problem. %
The reduction increases the number of edges quadratically, so we
cannot use it. %
Instead, we use the reduction in~\cite[Theorem~9.5]{Papa} and compose
it with a $3$-stretch to make the graph simple. %
The reduction is from \pp{\#NAE-$3$-Sat}$^+$ to \pp{\#MaxCut}.
\begin{lemma}%\hspace{4pt}
  There is a polynomial-time reduction from \pp{\#NAE-$3$-Sat}$^+$ to
  \pp{\#MaxCut} that maps formulas with $m$ clauses to graphs with
  $O(m)$ edges.
\end{lemma}
\begin{proof}
  We use the same reduction as~\cite[Theorem~9.5]{Papa} and we repeat
  the details here for completeness. %
  Given an instance~$\varphi$ of \pp{NAE-$3$-Sat} with $n$ variables and
  $m$~constraints, we construct a graph $G$ as follows:
  For every variable~$x_i$, we add adjacent vertices~$x_i$
  and~$\neg x_i$.
  For every constraint $\{a,b,c\}$ of~$\varphi$, we further add a
  triangle between the three involved literals, which possibly leads
  to multiedges.
  This multigraph $G$ has $2n$ vertices and $3m+n$ edges.

  With $k=2m+n$, we claim that the number of cuts of size~$k$ is equal
  to the number of not-all-equal assignments of $\varphi$.
  First notice that there are no cuts of size larger than~$k$:
  every constraint triangle either contributes zero or two edges to
  any cut~$C$, so every cut has at most $2m$ edges from constraint
  triangles of $G$.
  Except for triangle edges, there are exactly $n$ further edges in
  the graph, so the cut cannot be larger than $2m+n=k$.
  Also note that if any $x_j$ and $\neg x_j$ are on the same side of a
  cut, then the size of that cut cannot exceed $k-1$.
  Hence every cut~$C$ of size exactly $k$ separates all pairs~$x_i$
  and~$\neg x_i$ and can be seen as a truth assignment to the
  variables of~$\varphi$.
  Furthermore, since~$C$ has size exactly $k$, it cuts every
  constraint triangle, so it corresponds to a not-all-equal truth
  assignment of~$\varphi$.
  For the other direction, any cut constructed from a not-all-equal
  assignment separates all $x_i$ and $\neg x_i$, and cuts every
  triangle, so the size of such cuts is~$k$.
  In particular, since we reduced from an instance~$\varphi$ that has
  at least one not-all-equal assignment, the maximum cuts of~$G$ have
  size~$k$.
  We obtain a parsimonious polynomial-time reduction from
  \pp{\#NAE-$3$-Sat}$^+$ to \pp{\#MaxCut} on multigraphs that increases
  the parameters $n$ and $m$ at most by a constant factor.

  We now reduce \pp{\#MaxCut} for multigraphs to simple graphs.
  Let~$G$ be a multigraph with $m$ edges and with a maximum cut of
  size~$k$.
  Let $G'$ be the $3$-stretch of~$G$, that is, every edge is replaced
  by a path with three edges.
  This graph has $3m$ edges, and we claim that
  $\pp{\#MaxCut}(G')=3^{m-k} \cdot \pp{\#MaxCut}(G)$,
  which suffices to prove the reduction.

  To prove the claim, let $C$ be a maxcut of $G$.
  We think of $C$ as a colouring $C:V(G)\to\{0,1\}$ such that the
  number of bichromatic edges is maximized.
  The colouring $C$ can be extended in $3^{m-k}$ ways to a maximum cut
  of $G'$ as follows.
  We consider an edge $\{u,v\}$ of~$G$ that got stretched into a
  $3$-path $u,a,b,v$.
  \begin{enumerate}[(1)]
    \item If $C(u)=C(v)$, then there are exactly three ways to colour
    $a$ and $b$ such that the number of bichromatic edges on the path
    $u,a,b,v$ is two.
    Furthermore, no extension can yield more than two bichromatic
    edges.
    \item If $C(u)\neq C(v)$, then there is exactly one way in which
    colouring can be extended to $a$ and $b$ such that the number of
    bichromatic edges on the path $u,a,b,v$ is three.
  \end{enumerate}
  Since $C$ has $k$ bichromatic edges and $m-k$ monochromatic edges in
  $G$, it can be extended in $3^{m-k}$ ways to yield a colouring of
  $G'$ with $2(m-k)+3k=2m+k=k'$ bichromatic edges.
  On the other hand, any other extension than the above, as well as
  any extension of cuts $C$ of size smaller than $k$ lead to cuts of
  $G'$ that have size smaller than $k'$.
\end{proof}

\subsection*{Minimum cut between three terminals}
For convenience, we restate the definition of
\pp{\#$3$-Terminal MinCut} from \S\ref{sec: hyperbolas}.
\begin{quote}\small
\begin{description}
  \item[Name] $\pp{\#$3$-Terminal MinCut}$
  \item[Input] Simple undirected graph $G=(V,E)$ with three
    distinguished vertices (``terminals'') $t_1,t_2,t_3 \in V$.
  \item[Output] The number of cuts of minimal size that separate $t_1$
    from $t_2$, $t_2$ from $t_3$, and $t_3$ from $t_1$.
\end{description}
\end{quote}
\begin{lemma}
  There is a polynomial-time reduction from the \pp{\#MaxCut} problem to
  \pp{\#$3$-Terminal MinCut} that maps graphs with $m$ edges to graphs
  with $O(m)$ edges.
\end{lemma}
\begin{proof}
  We follow the reduction of Dahlhaus et al.~\cite[Theorem 3]{DJPSY94}. %
  So let $G=(V,E)$ be a simple graph with $n$ vertices and $m$
   edges. %
  It is made explicit in~\cite{DJPSY94} that the construction builds a
  graph $F$ with $n'=3+n+4m= O(m)$ vertices. %
  For the number of edges, every $uv\in E$ results in a gadget graph
  $C$ with $18$ edges, so the number of edges in $F$ is $18m=O(m)$. %
  The construction is such that the number of minimum $3$-terminal
  cuts of $F$ equals the number of maximum cuts of~$G$.
\end{proof}

\subsection*{Three-colouring}
\begin{quote}\small
\begin{description}
  \item[Name] $\pp{\#$3$-Colouring}$
  \item[Input] Simple undirected graph $G$.
  \item[Output] The number of proper vertex-colourings with three
    colours.
\end{description}
\end{quote}
\citeN{IPZ01} already observed the hardness of \pp{$3$-Colouring}
under \ETH{}.
This can be extended to the counting version as follows.
\begin{lemma}\hspace{-5pt}\label{NAE-$3$-Sat to 3-Colouring}%
  There is a polynomial-time reduction from the \pp{\#NAE-$3$-Sat} problem to
  \pp{\#$3$-Colouring} that maps formulas with $m$ clauses to graphs
  with $O(m)$ edges.
\end{lemma}
\begin{proof}
  We follow the proof of \cite[Theorem~9.8]{Papa}.
  The graph~$G$ that is constructed from an
  \pp{NAE-$3$-Sat}-instance $\varphi$ with~$n$ variables and~$m$ clauses
  has $n'=1+2n+3m$ vertices and $m'=3n+6m$ edges. %
  Furthermore, every not-all-equal assignment to the variables of
  $\varphi$ gives rise to exactly $3\cdot 2^m$ proper $3$-colourings
  of $G$:
  There are $3$ possible colours for $a$ and a variable assignment
  then uniquely colours the $2n$ vertices that correspond to literals
  (take the smaller of the remaining colours to mean false and the
  larger to mean true; since complements of not-all-equal assignments
  are also not-all-equal assignments, this choice prevents
  overcounting).
  Now the colouring can be extended to each clause gadget in exactly
  two ways.
  Hence the number of proper $3$-colourings of~$G$ is equal to
  $3\cdot 2^m\cdot \pp{\#NAE-$3$-Sat}(\varphi)$.
\end{proof}

\begin{proof}[of Theorem~\ref{thm: standard reductions}]
  Assume one of the problems can be solved in time $\exp(cm)$ for
  every $c>0$. %
  Then \pp{\#$3$-Sat} can be solved by first applying the applicable
  reductions of the preceding lemmas and then invoking the assumed
  algorithm. %
  This gives for every $c>0$ an algorithm for \pp{\#$3$-Sat} that runs
  in time $\exp(O(c m))$, which implies that \cETH{} fails.
\end{proof}

\subsection{Hardness of Colouring and Other Individual Points on the Chromatic Line}
\label{sec: Linial}

Theorem~\ref{thm: Tutte main result}\iref{thmi: Tutte linial}
cannot be handled by the proof of
Proposition~\ref{prop: individual points, multigraphs, nonzero q} because
thickenings do not produce enough points for interpolation. %
Instead, we use a reduction for the chromatic line that was discovered
by~\citeN{L86}.

The chromatic polynomial $\chi(G;q)$ of $G$ is the polynomial in $q$
with the property that, for all $c \in \N$, the value $\chi(G;c)$
is the number of proper $c$-colourings of the vertices of~$G$. %
We write~$\chi(q)$ for the function~$G \mapsto \chi(G;q)$. %
The Tutte polynomial specializes to the chromatic polynomial for
$y=0$:
\begin{equation}\label{Tutte chromatic}
  \chi(G;q) = (-1)^{n(G) - \comp(G)}q^{\comp(G)} T(G;1-q,0)\, .
\end{equation}
The following two propositions establish
Theorem~\ref{thm: Tutte main result}\iref{thmi: Tutte linial}.
\begin{proposition}\label{prop: Tutte linial}%
  Let $x\in \{-2, -3, \ldots\}$. %

  If \cETH{} holds, then $\pp{Tutte}^{0,1}(x,0)$ cannot be computed in
  time $\exp(o(m))$.
\end{proposition}

\begin{proof}
  Set $q=1-x$. %
  Since $q\neq 0$, it follows from \eqref{Tutte chromatic} that
  evaluating $\pp{Tutte}(x,0)$ is equivalent to evaluating the
  chromatic polynomial $\chi(q)$ at point $q$. %
  In particular, $\chi(3)$ is the number of $3$-colourings.
  By Theorem~\ref{thm: standard reductions}, if \cETH{} is true,
  $\chi(3)$ cannot be computed in time $\exp(o(m))$ even for simple
  graphs.
  For $i\in \{1,2,\ldots\}$ and all real $r$, Linial's identity is
  \begin{equation}\label{eq: linial}
    \chi(G+K_i;r)
    = r (r-1) \dots (r-i+1) \cdot \chi(G;r-i)\,,
  \end{equation}
  where $G+K_i$ is the simple graph consisting of $G$ and a
  clique~$K_i$ on~$i$ vertices, each of which is adjacent to every
  vertex of~$G$.

  For $q\in\{4,5,\ldots\}$, we can set $i=q-3$ and directly compute
  $\chi(G;3)=\chi(G;q-i)=\chi(G+K_i;q)/[q(q-1)\cdots 4]$. %
  Since $m(G+K_i)=m(G)+i \cdot n(G) + \binom{i}{2} \leq O(m(G))$, it
  follows that $\chi(q)$ cannot be computed in time $\exp(o(m))$ under
  \cETH{}, even for simple graphs.  %
\end{proof}

\begin{proposition}
  Let $x \notin \Q\setminus\{1,0,-1,-2, -3, \ldots\}$. %

  If \cETH{} holds, then $\pp{Tutte}^{0,1}(x, 0)$ cannot be computed in
  time $\exp(o(n))$.
\end{proposition}
\begin{proof}
  Set $q=1-x$. %
  We show that $\pp{Tutte}^{0,1}(x,0)$ cannot be computed in time
  $\exp{(o(n))}$ under \cETH{}. %
  Indeed, with access to $\chi(q)$, we can compute $\chi(G;q-i)$ for
  all $i=0,\dots,n$, noting that all prefactors in \eqref{eq: linial}
  nonzero. %
  From these $n+1$ values, we interpolate to get the coefficients of
  the polynomial $r\mapsto\chi(G;r)$, which in turn allows us evaluate
  $\chi(G;3)$. %
  In this case, the size of the oracle queries depends non-linearly on
  the size of $G$, in particular $m(G+K_n) \sim n^2$. %
  However, the number of vertices is $n(G+K_i) \leq 2n \leq O(m(G))$.
  Thus, since $\chi(3)$ cannot be computed in time $\exp(o(n))$ under \cETH{},
  this also holds for $\chi(q)$, even for simple graphs.
\end{proof}

The only points on the $x$-axis not covered here are $x\in\{1,0,-1\}$. %
Two of these admit polynomial-time algorithms, so we expect no
hardness result. %
By Theorem~\ref{thm: Tutte main result}\iref{thmi: Tutte reliability},
the Tutte polynomial at the point $(1,0)$ cannot be evaluated in time
$\exp(o(m/\log^2 m))$ under \cETH. %

\end{document}